\documentclass[sigplan,10pt]{acmart}
\setcitestyle{nocompress}   

\usepackage{microtype}
\usepackage{tikz}
\usepackage{booktabs}
\usepackage{array}
\usepackage{float}
\graphicspath{{figures/}}

\setcopyright{none}
\renewcommand\footnotetextcopyrightpermission[1]{}
\AtBeginDocument{%
  \fancypagestyle{standardpagestyle}{\fancyhf{}%
    \fancyfoot[C]{\fontsize{10}{12}\selectfont\thepage}}%
  \fancypagestyle{firstpagestyle}{\fancyhf{}%
    \fancyfoot[C]{\fontsize{10}{12}\selectfont\thepage}}%
  \pagestyle{standardpagestyle}%
  \setlength{\footskip}{30pt}}

\theoremstyle{plain}
\newtheorem{theorem}{Theorem}
\newtheorem{lemma}{Lemma}[section]
\newtheorem{corollary}{Corollary}[section]
\theoremstyle{definition}
\newtheorem{definition}{Definition}[section]

\newcommand{\Src}{\mathsf{Src}}            
\newcommand{\replay}{\mathsf{replay}}      
\newcommand{\scope}{K}                     
\newcommand{\fin}{\mathit{fin}}            
\newcommand{\corein}{\mathsf{core}}        
\newcommand{\gen}{\mathsf{gen}}            
\newcommand{\pay}{\mathsf{pay}}            
\newcommand{\stampf}{\mathsf{stamp}}       
\newcommand{\just}{\mathsf{justified}}
\newcommand{\prem}{\mathsf{premature}}
\newcommand{\dupl}{\mathsf{duplicate}}
\newcommand{\unsafe}{\mathsf{unsafe}}
\newcommand{\lostp}{\mathsf{lost}}
\newcommand{\alo}{\mathsf{alo}}            
\newcommand{\eo}{\mathsf{eo}}              
\newcommand{\clean}{\mathsf{clean}}        
\newcommand{\dedupv}{\mathsf{dedup}}       
\newcommand{\view}{\mathsf{view}}          
\newcommand{\dlv}{\mathsf{dview}}          
\newcommand{\rst}{\mathbin{\restriction}}  
\newcommand{\Running}{\textsf{Running}}
\newcommand{\Crashed}{\textsf{Crashed}}
\newcommand{\Recovered}{\textsf{Recovered}}
\newcommand{\rdrto}[2]{\mathrel{\rightsquigarrow_{#1\langle #2\rangle}}} 
\newcommand{\sdelta}{\Delta}               

\let\origtexttt\texttt
\renewcommand{\texttt}[1]{{\origtexttt{\def\_{\textunderscore\allowbreak}#1}}}
\usepackage{ragged2e}
\AtBeginEnvironment{thebibliography}{\RaggedRight}

\makeatletter
\@ACM@balancefalse
\makeatother
\usepackage{flushend}

\makeatletter
\def\@textbottom{\vskip\z@\@plus 9\p@}
\makeatother

\urldef{\kafkaexpurl}\url{https://kafka.apache.org/43/configuration/broker-configs/#brokerconfigs_producer.id.expiration.ms}

\title{Machine-Checked Dual-Write Recovery from a~Commit Log}
\author{Andreas Andreakis}


\begin{document}

\begin{abstract}
Applications often need to make related facts durable in two
independent systems without a transaction spanning both. If the process
crashes after the second system accepts an operation but before a
source-side checkpoint is written, recovery cannot tell from source
state alone whether to retry. Retrying may duplicate the effect, while
skipping may leave committed work incomplete. Transactional outboxes
and change data capture move this dual write out of an application
process, but relay delivery and checkpointing remain separate durable
operations.

The engineering problem is familiar, and systems address it
with retries, checkpoints, idempotency keys, and fencing, and call the
result exactly-once delivery. Whether that guarantee holds depends on
which event it counts, what evidence recovery requires, and how long
that evidence must survive. Formal
verification has covered adjacent problems, including transaction
isolation, crash safety within one store, and protocols with a shared
commit. The closest formal studies use model checking to verify
particular outbox and log-delivery designs, and their results hold
only for the designs and instances they check. Answering the
three questions in general rather than for one design requires
statements about arbitrary recovery policies, which finite enumeration
cannot reach. We prove them in Isabelle/HOL, and to our knowledge they
have not been machine-checked before.

The main result is an impossibility theorem for source-only recovery. We construct two reachable
post-crash states with the same durable source-side state and different
sink acceptance records. Any recovery policy based only on the source
side must duplicate an effect in one state or leave it undelivered in
the other. The same result holds for a deterministic
deliver-then-checkpoint protocol whose only nondeterminism is crash
timing. When the sink exposes an authoritative, complete, and current
acceptance record, recovery can compute exactly which committed
operations remain undelivered, provided source coordinates distinguish
operations. We extend the model to requests that remain in flight and
to concurrent recoverers, and prove an arrival fence and a claim fence
for these hazards.
Finally, we show how bounded deduplication state and truncated source
history limit the lifetime of the guarantee. A recovery design must
therefore identify the sink's acceptance event, prevent stale actors
from changing the result after it is read, and retain the required
evidence for the full recovery horizon.
\end{abstract}

\maketitle

\section{Introduction}\label{sec:intro}

Consider an order service that commits \textsf{shipped} to its
database. A relay sends a confirmation email for each committed order
and records a checkpoint after every send. The process then crashes and
restarts. Its checkpoint is now one order behind the commit log, so
recovery sends the last confirmation again and the customer receives
two copies. An inspection finds no error in the database as the rows are
correct, the checkpoint is valid, and the store-level audit
passes. This paper explains how a process that checkpoints after every
send can still fail to provide exactly-once delivery.

The failure occurs because delivery and checkpointing are separate
durable operations. The crash may occur after the mail provider
accepts the request but before the process records the checkpoint. At
restart there are two possible states. In one state the provider
accepted the request, and in the other it did not. The service can read
the same database, write-ahead log, outbox, acknowledgement log, restart counter, 
and checkpoint in both states. Sending the message again
creates a duplicate in the first state. Skipping it leaves committed
work incomplete in the second. The local checkpoint is not incorrect.
It simply does not record whether the receiving side accepted the
effect. Only a durable record at the receiving side can distinguish
the two states.

\paragraph{Dual writes.}
The incident is an example of a \emph{dual write}. Related facts
become durable under independent authorities without one commit
spanning both. The source may be a database, while the second
authority may be a mail provider, broker, cache, search index, or
another service. The model distinguishes a durable source commit from
a durable sink acceptance. Section~\ref{sec:model} and the
theorem statements give the remaining requirements.

The usual remedy is to derive the second write from the source's
committed log through change data capture (CDC) or a transactional
outbox~\cite{kleppmann2015dualwrites,ddia,kleppmann2019olep}. This
approach removes the dual write from application code.
Section~\ref{sec:store} proves what it provides at the store tier. A
derived copy is crash-safe exactly when it matches the committed log
up to the log position it advertises. The relay must still deliver
each derived event and record its own progress. These are again two
durable operations without a shared commit. The log-based approach 
therefore moves the recovery problem to the relay. The replayable source
log lets the relay retry failed deliveries and catch up after an outage.
Recovery still depends on the acceptance evidence and controls provided
by the sink.

\paragraph{Why formalize now.}
The engineering guidance for this problem is established. Prior work
recommends deriving state from a log, retrying failed deliveries, using
stable identities, and fencing stale
workers~\cite{kleppmann2015dualwrites,ddia,treat,kleppmann2019olep}.
Systems that combine these mechanisms call the result exactly-once
delivery. Yet such claims often leave unclear what counts as delivery,
how recovery can determine whether it occurred, and how long the required
evidence remains valid.
This can make an information limit look like a checkpoint bug. It
can also make a conditional guarantee look impossible when the
required sink-side evidence and controls are available.

Formal verification has addressed nearby problems such as transaction
isolation~\cite{veriso}, single-store crash
recovery~\cite{gojournal,perennial}, and distributed protocols with a
shared commit~\cite{disel}. Those models do not decide whether an
independent sink already accepted an effect. The closest formal
studies use bounded model checking for outbox and log-delivery
designs~\cite{masternak,logplayer}. To our knowledge, this is the first machine-checked proof of this
recovery limit for all policies that use only the modeled source-side
state.

\paragraph{Results.}
The results are organized around three parts of the recovery problem.

\emph{Source-side impossibility result.}
Sections~\ref{sec:wall} and~\ref{sec:checkpoint} construct reachable
post-crash states that agree on the durable information available to
recovery and differ in the sink's accepted record. Every policy that
reads only the source side must choose the same recovery batch in both
states. That batch duplicates a delivery in one state or leaves one
undelivered in the other. Section~\ref{sec:checkpoint} strengthens the result
to one deterministic deliver-then-checkpoint protocol. The protocol
maintains its cursor correctly, and crash timing alone creates the two
indistinguishable local states.

\emph{Stale recovery information.}
Section~\ref{sec:door} proves that recovery can compute the missing
batch from an authoritative, complete, and current sink acceptance
record under stated conditions. Sections~\ref{sec:wire}
and~\ref{sec:claim} then identify two ways in which that read can
become stale. A request from the crashed process may still be in
flight, or another recoverer may act on the same crash. We prove an
arrival fence for the first hazard and a claim fence for the second.
Both mechanisms act at the sink's acceptance boundary. The arrival
fence also has a proved cost because it may reject an old request that
would otherwise have completed later work.

\emph{Evidence lifetime.}
Section~\ref{sec:memory} proves the standard relationship between
at-least-once delivery and a permanently deduplicated sink view. It
then shows what changes when deduplication state expires or the source
history is truncated. Once the sink forgets an accepted identity, a
valid replay may be processed as new. Once recovery loses the content
of the source history, it may no longer be able to distinguish work
that requires delivery from work that was never committed. The guarantee
therefore lasts only as long as the evidence required by the recovery
procedure.

\paragraph{Scope.}
This paper presents a theoretical study of dual-write recovery.
All theorems are machine-checked in Isabelle/HOL. The models, theorem
statements, assumptions, and proof ideas are included in the text, so
the paper can be read without the formal artifact. The negative results
are constructions over specific machines and policy classes. They show
that the stated failure is reachable and that every policy in the class
fails on the construction, without claiming that every deployed system fails.
The positive results hold under conditions stated with each theorem. We do
not verify a broker, connector, mail provider, or deployed recovery implementation.

The model covers one source authority, one accepting endpoint, and
one delivery direction, alongside in-flight delivery and concurrent
recovery on separate machines. Section~\ref{sec:mech} lists further
boundaries. The theory aims to identify the evidence and acceptance
controls required by a recovery claim, not to certify a product.

\paragraph{Organization and reading paths.}
The paper provides two complementary reading tracks.
Readers focused on formal foundations can follow the
\emph{Theory Track} (Sections~\ref{sec:model}-\ref{sec:related}),
which develops the transition systems, impossibility results,
fence constructions, and retention limits.
Readers focused on practical systems can take the
\emph{Practitioner Track}, jumping directly to
Section~\ref{sec:practitioner}. That self-contained guide requires no
proof mechanics and translates the results into operational guidance for
direct dual writes and CDC outbox relays, complete with sink archetypes,
hazard analyses, and recovery tests.
Section~\ref{sec:conclusion} concludes.

\section{The Model}\label{sec:model}

All results use a common state model. This section defines the shared
state, the first transition system, and the terminology used throughout
the paper. Later sections extend the state with a checkpoint, an
in-flight channel, concurrent recovery, retention, and the store tier.

\paragraph{Histories, frontiers, materialization.}
A \emph{source event} $e$ on a key $k$ inserts or updates that key to a
value, or deletes it. A committed \emph{source history} $H$ is a finite
sequence of (coordinate, event) pairs whose coordinates are
non-decreasing. A \emph{frontier} $f$ is a coordinate of that
history, never a wall-clock time, a broker offset, or a sink-local
sequence number. Given a base map $b$, the \emph{materialization}
$\Src(b,H,f)$ gives each key the value of its latest event at a
coordinate at most $f$ (a delete yields absence), falling back to $b$
where the key has no such event. Formally, histories are arbitrary
lists of this shape and frontiers arbitrary coordinates. The
non-decreasing wellformedness just described, and membership in a
machine's reachable executions, are imposed by stated assumptions and by
the machines' own transition rules, not by the raw data type. The
\emph{core execution state}
\[
\sigma = (b, H, D, Q, \mathit{pend}, \scope, \fin, \mathit{st},
\mathit{Ack})
\]
carries the base, the committed history, an independently delivered
downstream history $D$, an enqueued history $Q$ with the pending set
$\mathit{pend}$ of enqueued-but-undelivered events, the key scope
$\scope$, the interval end $\fin$, a run status
$\mathit{st} \in \{\Running,\ \Crashed\,c,\ \Recovered\}$, and a
source-side acknowledgement history $\mathit{Ack}$. Labelled steps
drive the state: commit to $H$, enqueue, deliver to $D$, acknowledge
into $\mathit{Ack}$, crash (which freezes the run at a
frontier), and recover. A completed recovery cycle returns the run to
\Running{} with a new generation. A two-event instance grounds the
notation and returns throughout the paper: with
$H = [(c_1, \mathsf{ins}\,k\,7),\ (c_2, \mathsf{upd}\,k\,9)]$ and an
empty base, $\Src(b,H,c_1)(k) = 7$, $\Src(b,H,c_2)(k) = 9$, and the
scoped obligations at $c_2$ are both events, in order.

Both stores are read through the one materialization operator at the
same source frontier: $\Src(b,H,f)$ is what the committed log says,
$\Src(b,D,f)$ is what the downstream holds.

\begin{definition}[Scoped mismatch]\label{def:mismatch}
At frontier $f$, the downstream \emph{mismatches} the committed log on
key $k$ when $k\in\scope$ and $\Src(b,D,f)(k) \neq \Src(b,H,f)(k)$.
Write $\clean(\sigma,f)$ when no scoped mismatch holds at $f$.
\end{definition}

When a crash can freeze such a mismatch where an observer sees it is a
question about the \emph{store tier}, and it has an exact answer,
proved in Section~\ref{sec:store}. The subject of this paper sits
above that tier, at a boundary the store never records.

\paragraph{The acceptance boundary.}
The \emph{effect machine} extends $\sigma$ with one field and one
register. A state is $t = \langle \sigma, E, \gamma \rangle$: the core,
an \emph{emissions ledger} $E$, and a \emph{generation} $\gamma$. The
ledger records received effects such as emails sent, messages
published, and webhook calls fired. It is
append-only by construction: a received effect cannot be unreceived,
so no step of any machine in this paper removes a ledger entry. No
store field can read the ledger. It models what happened,
not what the crashed side knows. The generation $\gamma$ is the restart counter of
Section~\ref{sec:intro}. It identifies successive process runs, starting at zero. Each completed resume enters
the successor generation. Related counters in deployed systems are called
epochs or terms.

An \emph{emission} is a triple $x = \langle n, g, (c,e) \rangle$: the
\emph{stamp} $n$ (the length of the committed history when it
fired), the generation $g$ current when it fired, and the \emph{payload}
$\pay(x) = (c,e)$. An emission is \emph{justified} when the durable
committed prefix covered it at fire time, a check the stamp makes
readable from the state:
\[
  \begin{aligned}
  \just(H, x) \;\longleftrightarrow\;&
    \stampf(x) \le |H|\\[-2pt]
  &{}\wedge\; \pay(x) \in \mathrm{take}(\stampf(x), H).
  \end{aligned}
\]
For a ledger $L$, $\pay(L)$ is the ordered pointwise projection
$[\,\pay(x)\mid x\leftarrow L\,]$. On this first machine, delivery is
instantaneous: an ordinary downstream delivery advances $D$ and appends
the received effect to $E$ in a single transition. Section~\ref{sec:wire} splits
that collapsed send/receive step, because a whole class of hazards
arises in the gap between the two.

\paragraph{Hazards and missing deliveries.}
Three predicates over $t$ classify how a ledger can be wrong, and one
checks for missing deliveries. The \emph{obligations} at a frontier are the
scoped, frontier-bounded committed events
\[
  \replay(H,\scope,f) \;=\;
  [\, (c,e) \leftarrow H \mid c \le f \,\wedge\, \mathit{key}(e) \in \scope \,],
\]
and the hazard vocabulary is
\[
\begin{aligned}
  \prem(t) \;&\longleftrightarrow\;
     \exists x \in E .\; \neg\, \just(H, x),\\
  \dupl(t) \;&\longleftrightarrow\;
     \text{some payload repeats on } E,\\
  \unsafe(t) \;&\longleftrightarrow\; \prem(t) \vee \dupl(t),\\
  \lostp(f,t) \;&\longleftrightarrow\;
     \exists p \in \replay(H,\scope,f) .\; p \notin \pay(E),\\
  \alo(f,t) \;&\longleftrightarrow\;
     \replay(H,\scope,f) \subseteq \pay(E).
\end{aligned}
\]
$\unsafe$ is a safety verdict over what \emph{was} emitted: something
fired too early, or fired twice. It does not constrain what never
fired.
$\lostp$ and $\alo$ are each other's complements and form the
completeness axis. Read at a state, they say whether a frontier still
has an undelivered payload. They make no liveness claim. They classify the
state where a result is judged, not what some later
execution might still deliver. The theorems apply them at
post-recovery states, where the recovery being judged has finished its
work. The verdicts are order-invariant. What was delivered matters,
but its order does not. Multiplicity matters to $\dupl$ by
definition, and to the completeness pair only under the
ascending-coordinate condition P4 (Section~\ref{sec:door}).
``Exactly-once at $f$'' always means the conjunction (no hazard, and
at-least-once at $f$), and the bare phrase is never used as a formal
name in this paper: every use is frontier-relative and counted as its
theorem states.

The core's acknowledgement history $\mathit{Ack}$ records
\emph{source-side} acknowledgements, and it supports no delivery
conclusion: equal acknowledgement histories can coexist with different
emission ledgers. An acknowledgement is what the caller heard, not
what the other side did. The opening incident turned on this
distinction, and the model makes it structural.

\paragraph{What each side must provide.}
The model abstracts over both ends. On the source side,
the analysis assumes only a durable, ordered, append-only
committed history: each fact either is in $H$, with a stable
coordinate, or is not. How a fact became durable lies outside the
model: no relational store, multi-object transaction, or isolation
discipline is assumed anywhere. A write-ahead log, an event store, a
replicated queue, or a store limited to single-document writes each
yields an $H$. Strictly ascending coordinates are required twice
(condition P4, Section~\ref{sec:door}). They keep the sink-delta
re-drive batch duplicate-free, and they lift the counting results from
set level, by payload, to per instance. A source that can
lose committed content is not left unmodeled: this section's machines
never shrink $H$, and truncation gets its own regime in
Section~\ref{sec:memory}. And a fact with no single
committing write has no place in $H$ until one write is designated as
its durability event. A fact spread across independent writes with no
such point reproduces this paper's problem one level down, inside the
source.

On the sink side, the judged event is \emph{durable acceptance}, the
moment the receiving authority durably admits the operation: a
broker's acknowledged append, a mail provider accepting a send
request, a store committing a put. Nothing in this paper reaches past
acceptance, not into internal processing, onward delivery, or what a
user interface eventually shows. Sinks differ in what they let
recovery \emph{read}, not in their kind. A sink exposing its own
durable, complete per-operation accepted record supplies the raw
material for the positive results ahead. Their conditions are stated where
they are proved, starting in Section~\ref{sec:door}. A sink readable
only as current contents gets strictly less, and what that read can
and cannot separate is itself proved (Sections~\ref{sec:wall}
and~\ref{sec:checkpoint}). The observation bounds in
Sections~\ref{sec:wall}-\ref{sec:checkpoint} apply when recovery can
read only the modeled source-side state. Source acknowledgements do
not distinguish the constructed cases. The obligation runs one way, from the source's
record to the sink's acceptance. Obligations running both ways are reached here only one
direction at a time, and if even the source keeps no such record, the
question cannot be posed at all.

\paragraph{Re-drives, policies, observations.}
A \emph{re-drive policy} $P$ maps a state to a batch of payloads. The
re-drive relation $t \rdrto{P}{f} t'$ is one atomic recovery step doing
two things: it heals the store (the delivered history is reconciled to
the committed image at $f$, so $\clean(t',f)$ holds after \emph{every}
re-drive, whatever the policy), and it appends $P(t)$'s emissions to
the ledger, stamped with the current history length and generation.
That the step is one atom is a deliberate model boundary.
Section~\ref{sec:mech} discusses this atomicity assumption, and
Section~\ref{sec:door} shows that the central bound does not lean on
it. An \emph{observation} is any function $\mathit{obs}$ of the state,
of any result type. A policy $P$ \emph{factors through} $\mathit{obs}$ when
$P = g \circ \mathit{obs}$ for some $g$. Factoring is the paper's only
notion of what recovery reads: a policy that factors through
$\mathit{obs}$ can use $\mathit{obs}$'s output and nothing else. We consider three observations. The complete modeled store tier includes
both source and downstream histories, with
$\mathit{obs}(t) = \corein(t)$. The second observation adds the
generation, $\mathit{obs}(t) = (\corein(t), \gen(t))$. The third also
includes the sink's acceptance record. Section~\ref{sec:wall} rules
out policies based on the first two observations. Section~\ref{sec:checkpoint}
then studies a fixed protocol using the durable-local view, which
excludes downstream content as well as the acceptance record.
Section~\ref{sec:door} gives a correct sink-reading policy under its
stated conditions.

\paragraph{Mapping the model to systems.}
Table~\ref{tab:subst} fixes the substitution between practitioner
vocabulary and model objects. Two rows distinguish roles. An outbox
table is the committed log when a relay reads it
as its source, and it is an accepted record only if it also holds
sink-owned acceptance state updated atomically with acceptance. Merely
consulting a source outbox or a relay cursor does not turn either into
the sink's record. Which theorems govern a physical table depends on
the role its fields actually play. The substitution is one-way,
as theorems are stated over abstract model objects rather than concrete
technologies like webhooks or caches. The target-agnostic
objects subsume them. We avoid ``consistent'' because it conflates
several properties. When this paper means the stores agree at a
frontier it says mismatch-free, and
where it means delivery went right it names the hazard and the
frontier.

\begin{table}[t]
\caption{What the model's objects stand for. The mapping is one-way
(theorems name only the left column), and the outbox rows are
role-split on purpose.}
\label{tab:subst}
\footnotesize
\begin{tabular}{@{}>{\raggedright\arraybackslash}p{0.42\columnwidth}>{\raggedright\arraybackslash}p{\dimexpr0.58\columnwidth-2\tabcolsep\relax}@{}}
\toprule
\textbf{Model object} & \textbf{What it stands for} \\
\midrule
emission $x$ on the ledger $E$ & the second write: an email, a webhook
call, a cache fill, an index update \\
\addlinespace[2pt]
accepted record $A$ (\S\ref{sec:wire}) & the sink's durable record of
what it has accepted \\
\addlinespace[2pt]
wire $W$ (\S\ref{sec:wire}) & retries, socket buffers, requests still
in flight \\
\addlinespace[2pt]
store state inside $\sigma$ & source/downstream histories, queue and
pending state, run status, source acknowledgements \\
\addlinespace[2pt]
cursor (\S\ref{sec:checkpoint}) & a durable checkpoint of delivery
progress \\
\addlinespace[2pt]
generation $\gamma$ (\S\ref{sec:model}), fence $\varphi$
(\S\ref{sec:wire}) & producer
generation and the sink's acceptance threshold \\
\addlinespace[2pt]
committed log $H$ & the source's log, \emph{or an outbox table read
as the relay's source} \\
\addlinespace[2pt]
accepted record $A$ & sink-owned acceptance state colocated with an
outbox only when updated atomically with acceptance
(\S\ref{sec:store}) \\
\addlinespace[2pt]
journal $J$, floor $\mathit{fl}$ (\S\ref{sec:claim},
\S\ref{sec:memory}) & the upstream log as justification authority,
retention and compaction \\
\bottomrule
\end{tabular}
\end{table}

\paragraph{Scope of the results.}
Negative results concern the policy classes and schedule families
named in their statements. Positive results hold under their stated
conditions. The store-tier equivalence in Section~\ref{sec:store}
ranges over all implementations satisfying its assumptions.

\section{The Store-Side Observation Bound}\label{sec:wall}

After a crash, recovery must decide what to re-send. The first two
theorems show that no policy computed only from the durable store and
its control plane is correct on all reachable states. The state
available to recovery does not reveal whether the sink accepted the
operation.

\paragraph{The designed pair.}
We construct two reachable post-crash states with the same store
state and generation but different acceptance records. The policy
being tested chooses what to resend from these states. It need not
have produced the histories that led to them.
Section~\ref{sec:checkpoint} shows how crash timing creates this
ambiguity in a fixed protocol.

Two runs of the effect machine commit the same two facts and crash with
the second delivery enqueued but unsent. Their first recoveries then
diverge in one respect: the machine's recovery rule takes a
\emph{reconcile index} $m$ marking where re-delivery starts, and the
two runs use different values of it. With $m=1$, reconcile emits the
missing payload, and the emissions ledger records it. With $m=2$, reconcile
rebuilds the store and emits nothing, an advanced-cursor pathology one
notch away from the incident's stale-cursor duplicate. Both runs then
resume and stand crashed again at the same frontier $c_2$ with the
store fully caught up. At that point their formal cores are byte-equal:
base, committed and delivered histories, queue and pending state,
scope, interval, status, and source acknowledgements all agree, and so
do their generations. The earlier $m$ choices were parameters of
actions already taken. No trace of them remains in either state. Both
states are hazard-free, and both have genuinely emitted work
($E\neq[\,]$). One member still has an undelivered payload, and the other has
already delivered it. They differ only in the emissions ledger:
\[
  \pay(E(t_1)) = [(c_1,e_1),\,(c_2,e_2)],
  \qquad
  \pay(E(t_2)) = [(c_1,e_1)].
\]

\begin{theorem}[Observation Bound]\label{thm:obsbound}
On the effect machine there are reachable states $t_1, t_2$ and a
frontier $c_2$ with
\[
  \corein(t_1) = \corein(t_2), \quad
  \gen(t_1) = \gen(t_2), \quad
  E(t_1) \neq E(t_2),
\]
both hazard-free and with nonempty emission ledgers, such that for every
observation $\mathit{obs}$ of any result type with
$\mathit{obs}(t_1) = \mathit{obs}(t_2)$, and every policy
$P = g \circ \mathit{obs}$, there are re-drives
\[
  t_1 \rdrto{P}{c_2} t_1' \quad\text{and}\quad t_2 \rdrto{P}{c_2} t_2'
\]
with $\clean(t_1',c_2)$ and $\clean(t_2',c_2)$, and
\[
  \unsafe(t_1') \,\vee\, \lostp(c_2, t_2').
\]
\end{theorem}

Any policy that observes the same state in both executions must choose
the same batch. If the batch contains the missing payload, one execution
receives a duplicate. If the batch omits it, the other execution
remains incomplete. Adding uncommitted work only creates a premature
effect. Both repaired stores are mismatch-free, so a store audit
cannot distinguish the failure. Figure~\ref{fig:observation-fork}
shows this argument.

\begin{figure}[t]
\centering
\includegraphics[width=\columnwidth]{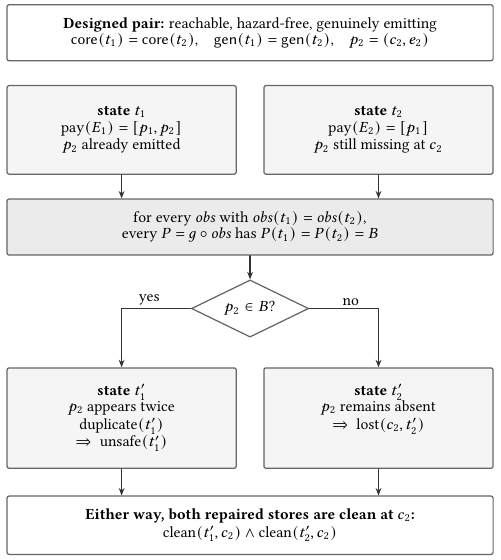}
\Description{Two reachable states with equal recovery observations
converge on one recovery batch. A decision diamond asks whether the
missing payload is in that batch. One branch duplicates it in the first
state and the other loses it in the second. Both branches end with
clean repaired stores.}
\caption{The constructed pair for
Theorems~\ref{thm:obsbound}-\ref{thm:controlplane}. The states have
equal core and generation but different emissions ledgers. Any common
observation yields the same batch, which duplicates $p_2$ in $t_1$ or
leaves it undelivered in $t_2$. Both repaired stores are mismatch-free.}
\label{fig:observation-fork}
\end{figure}

The observation is otherwise unrestricted. It may be any function of
the state, with any result type, and only equality on the constructed
pair is assumed. The result applies only to information represented
by the modeled observation.

\begin{theorem}[Control-Plane Bound]\label{thm:controlplane}
Every store-measured policy, computed as $P(t) = g(\corein(t))$
for every $t$, and every store-and-generation-measured policy, computed
as $P(t) = g(\corein(t), \gen(t))$, is defeated on a
reachable effect-machine pair as in Theorem~\ref{thm:obsbound}: both
re-drives leave the store mismatch-free at the pair's frontier, and the
outcome is a hazard on one member or a loss on the other.
\end{theorem}

\begin{proof}[Proof idea]
Policies of both classes agree on the designed pair, because the pair
is equal on core and generation. Theorem~\ref{thm:obsbound} then
applies. The mechanized proof checks the three-way split
directly (empty batch, batch containing a payload from the committed
prefix, non-empty batch entirely outside it), and each case yields one
of the stated outcomes.
\end{proof}

\paragraph{General and pair-specific outcomes.}
The disjunct in Theorem~\ref{thm:obsbound} is $\unsafe \vee \lostp$,
and since $\unsafe$ itself splits, the general result has three outcomes:
duplicate, lose, or fire prematurely. At the designed pair the
checked statement sharpens to the practitioner's two, because in the
premature case the skip-side member still lacks its delivery: the
verdict there is \emph{duplicate a delivery or leave it incomplete at the judged frontier}.
Both forms are machine-checked. The three-outcome statements remain
the general results, and the two-outcome readings apply to the designed
pairs used in the examples. The loss outcome carries
its standing disclosure: it is judged where the recovery under
examination has finished, not a claim that no later recovery could
still act.

\paragraph{Store and ledger health are independent.}
Store health and ledger health are orthogonal on this machine, in both
directions: states with a broken store and a clean ledger are
reachable, and so are states with a healed store and a broken ledger.
A second crash-and-re-drive cycle can raise the same payload's
count in the emissions ledger to three. An accurate delivered-count cursor
persisted \emph{outside} the
store would separate the pair, by being a sink-side record under
another name, which is precisely the observation class
Section~\ref{sec:door} builds on. Distinguishing the pair alone does not establish a recovery guarantee.
The positive results below apply to a specific policy under stated
conditions.

There is also a structural version of that boundary. A proved
dichotomy partitions all policy functions: every re-drive policy
either is a function of the modeled store and generation (and fails on the pair in
Theorem~\ref{thm:controlplane}) or distinguishes some two states
that agree there and differ on the ledger. The witnessing states of the
second class need not be reachable. The split is extensional, a
statement about information rather than operations. Its value is
exhaustiveness: at this abstraction there is no third kind of policy.

\paragraph{Scope of the observation bound.}
Both slogans that circulate about this problem fail against these
results. ``Exactly-once is impossible'' is too coarse:
Sections~\ref{sec:door},~\ref{sec:wire}, and~\ref{sec:memory} prove
exactly-once results at their frontiers, under conditions a system can
satisfy. The true statement is the bound above (no policy
blind to the sink's record achieves it), together with the conditions
under which the positive results ahead hold. ``Just read the sink before
re-driving'' fails in the other direction: once deliveries travel,
Section~\ref{sec:wire}'s straggler defeats every channel-blind batch
selector on the constructed pair under the unfenced re-drive relation.
The positive result changes acceptance instead of reading more.

\section{The Checkpoint Window}\label{sec:checkpoint}

The observation bound uses two earlier reconciliation choices. To show
the same ambiguity in a fixed implementation, we add a durable cursor
and study a deterministic deliver-then-checkpoint protocol. A crash
before delivery and a crash after delivery but before checkpointing
leave recovery with the same durable-local view and different sink
records. The cursor is part of the common view.

\paragraph{The checkpointed machine.}
Extend the effect machine with a single durable field: a progress
cursor. Delivery and checkpointing are separate
transitions: an ordinary delivery advances the downstream and the
ledger as before, and a separate \emph{persist} step advances the
cursor by one, permitted only when there is an actual delivery to
acknowledge. A crash may land between them, and that window is the
subject of this section.
The machine's one recovery rule is the effect machine's emitting
reconcile with its index $m$ \emph{instantiated by the machine's own
cursor}: recovery re-fires the committed suffix from wherever the
durable cursor points. The free parameter of
Section~\ref{sec:wall} is thereby instantiated by a durable cursor:
this machine's recovery action carries only the frontier, and nothing
else is left for an adversary to choose at recovery time.

\paragraph{Fixed protocol and crash schedule.}
On this machine we pin one \emph{protocol}: a function from states to
next actions implementing the faithful deliver-then-persist discipline
over the running two-event workload (commit $e_1$, enqueue $e_1$,
deliver $e_1$, and persist, then repeat these four actions for $e_2$
and stop). The protocol
itself never crashes. A run consumes a \emph{schedule} over a
two-letter alphabet: \textsf{Go} executes the protocol's unique next
action, \textsf{Crash} is the adversary's only move, the machine's
ordinary status-flip crash, at the point in the run where it lands. The
endpoint of a run is a function of its schedule: crash placement is the
only nondeterminism anywhere in the construction, and no per-run
parameter exists that could distinguish runs.

Two schedules matter: $\textsf{Go}^{7}\,\textsf{Crash}$ and
$\textsf{Go}^{6}\,\textsf{Crash}$ (the crash landing just after the
second delivery, versus just before it). Both runs end crashed with
cursor $1$: the first inside the second deliver-to-persist window, the
second right before that window opens. At the first run's endpoint, the
emissions ledger holds both payloads. At the second it holds only the first,
exactly the ledgers of Section~\ref{sec:wall}'s pair, this time
produced by one fixed protocol (Figure~\ref{fig:window}).

\begin{figure}[t]
\centering
\includegraphics[width=\columnwidth]{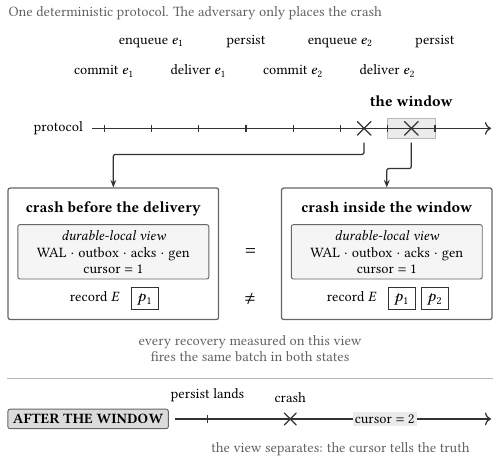}
\Description{A protocol timeline with eight labeled steps and a shaded
window between the second delivery and its persist. Arrows connect a
crash before the delivery and a crash inside the window to two cards.
The cards' durable-local view chips are identical and joined by an
equals sign, while their record cells differ, one holding one payload,
the other two. A lower rail shows a crash after the persist yielding
cursor two, which separates the view.}
\caption{The checkpoint window. Crashes immediately before the second
delivery and between that delivery and its persist produce equal
durable-local views but different emissions ledgers
(Theorem~\ref{thm:checkpoint}). A crash after persist advances the
cursor and separates the states.}
\label{fig:window}
\end{figure}

\paragraph{What recovery can read.}
Call the \emph{durable-local view} of a state everything a recovery
procedure of the crashed process reads without querying the sink side:
the base map, the committed source log, the enqueue (outbox) log, the
acknowledgement log, the scope, the interval end, the run status, the
recovery generation, and the cursor. Write $\dlv(\kappa)$ for that
projection, and call a policy \emph{durable-local-measured} when it
factors through it. Three parts of the full state sit outside the view,
each for a stated reason. The emissions ledger is the sink-side record.
Excluding it is the entire question. The downstream store's
delivered content is excluded because reading it \emph{is} reading the
sink: it is not local state, and the pair below is provably not
equal on it, so that read genuinely separates the members. The
dichotomy at this section's end places such readers outside the
defeated policy class. The pending delivery window is
volatile by this machine's own modeling assumption: it is the crashed
process's in-memory buffer, which recovery in this model does not
retain. For each exclusion, a companion result shows that adding the
excluded component back changes the
verdict. The crash transition keeps every field. The model restricts which
of those fields recovery may read after the crash.

\begin{theorem}[Checkpoint Dilemma]\label{thm:checkpoint}
On the checkpointed machine there are reachable states
$\kappa_1, \kappa_2$, the endpoints of the two protocol schedules
above, with
\[
  \dlv(\kappa_1) = \dlv(\kappa_2), \qquad
  E(\kappa_1) \neq E(\kappa_2),
\]
both hazard-free and with nonempty emission ledgers, such that every
durable-local-measured policy $P$ is defeated at the finish frontier
$c_2$: there are re-drives
$\kappa_1 \rdrto{P}{c_2} \kappa_1'$ and
$\kappa_2 \rdrto{P}{c_2} \kappa_2'$, both leaving the store
mismatch-free at $c_2$, with
\[
  \unsafe(\kappa_1') \,\vee\, \lostp(c_2, \kappa_2').
\]
At this pair the checked verdict again sharpens to two outcomes: duplicate
a delivery, or leave one undelivered at the judged frontier.
\end{theorem}

The cursor is inside the view: recovery reads it, and it is equal on
the two members. So are the committed log, the outbox, the
acknowledgements, the status, and the generation. The pair does not
arise from pathological bookkeeping. It is the deliver-then-persist
window, crossed by a crash, seen from the inside.

\paragraph{The deterministic protocol.}
The defeat quantifies over all durable-local-measured policies, and the
machine's own faithful recovery inhabits the class, so the theorem
applies to it. A checked trace also runs it concretely. Fired from
the two view-indistinguishable members, the same parameter-free
recovery transition is exactly right on the lagging member (store
healed, no hazard, no delivery missing at $c_2$) and produces an effect-unsafe
duplicate on the delivered member, every entry of which is individually
justified. That is the familiar at-least-once window of checkpointed
delivery, stated and checked as a theorem. Recovery emits $p_2$ after
either crash, but only the crash after delivery makes that emission a
duplicate.

\paragraph{Relation to the equal-core pair.}
One distinction separates this result from
Theorem~\ref{thm:obsbound}. That pair was equal on the whole modeled
core. This pair is equal on the durable-local view and provably
\emph{not} equal-core: the delivered member's downstream history and
pending window record the delivery, and the model exposes this
difference. The two shapes are complementary. The mechanization also
proves that before the first
resume, no reachable pair of crashed states has equal cores and
ledgers that differ in their delivered payloads, so an equal-core,
payload-divergent pair generated by one crash-timed protocol at
generation zero cannot exist. Theorem~\ref{thm:obsbound} is the
wider information bound over every observation that agrees on its
equal-core pair. Theorem~\ref{thm:checkpoint} is the protocol-specific
bound over the reads a crashed implementation actually has. The two theorems
therefore cover different recovery observations.

On this pair, reading the
\emph{downstream store's content} separates the members because it
contains the delivered copy, so a content reader distinguishes this
pair. That read is not store-local information. It
is a read of the sink, in its cheapest form. But it does
not survive Section~\ref{sec:wall}: on the equal-core pair the
downstream contents agree, so those reads do not distinguish the two states.
Only the sink's per-operation accepted record distinguishes
\emph{both} regimes: the acceptance history, not the current contents.
Reading the sink is therefore two different acts, reading its state
and reading its acceptance history, and the distinction returns as a
theorem boundary in Sections~\ref{sec:memory} and~\ref{sec:store}.

\paragraph{Policy classes and the cursor.}
As in Section~\ref{sec:wall}, the split is exhaustive: every policy on
this machine either factors through the durable-local view (and the
pair defeats it) or distinguishes some two states agreeing on that
entire view, which means it reads the downstream content, the volatile
window, or the ledger. The cursor itself is not the defect: if the
crash lands \emph{after} the
final persist, the cursor reads $2$, which by itself separates that
state's view from both pair members. Outside the deliver-to-persist
window the cursor is accurate. The dilemma is the window, not the
cursor's bookkeeping. That is why no amount of care in maintaining the
checkpoint, and no smarter encoding of it, closes a gap that sits
between two durable writes.

\section{Reading the Sink}\label{sec:door}

The pairs in the last two sections differ in the sink-side accepted
record. A policy that reads this record can compute what is still
missing. The \emph{sink delta} at frontier
$f$ is the committed obligations not yet on the record:
\[
  \sdelta\langle f\rangle(t) \;=\;
  [\, p \leftarrow \replay(H,\scope,f) \mid p \notin \pay(E(t)) \,].
\]
Here $E(t)$ is the model's authoritative, complete record at the
frontier of the decision. That assumption is necessary. A
stale replica, a paginated status endpoint, or a lagging projection of
the sink's state is a \emph{different} observation, and nothing below
applies to it without a separate argument. The guarantee is exact, and
it is conditional. The result requires four conditions:
\begin{itemize}
\item[P1.] \emph{Hazard-free:} $\neg\,\unsafe(t)$ (the record carries
  no duplicate and nothing premature).
\item[P2.] \emph{Crashed:} the run's status is $\Crashed\,c$ for some
  $c$.
\item[P3.] \emph{In-interval:} $f \le \fin$.
\item[P4.] \emph{Ascending:} the committed history's coordinates are
  strictly increasing.
\end{itemize}

\begin{theorem}[Sink-Reading Escape]\label{thm:escape}
On the effect machine, under P1-P4, the sink-delta re-drive at $f$
exists, and every re-drive $t \rdrto{\sdelta}{f} t'$ leaves the store
mismatch-free at $f$, hazard-free, and complete:
\[
  \clean(t',f), \qquad \neg\,\unsafe(t'), \qquad \neg\,\lostp(f, t').
\]
\end{theorem}

\begin{proof}[Proof idea]
Existence follows from the atomic reconcile. The three conclusions
follow from the delta's shape. It never re-emits a payload already on the record, so no
duplicate occurs. It emits only committed obligations, stamped at the
full history, so nothing premature occurs. It includes every payload still missing
at $f$, so no loss remains. None of the three is automatic: each depends on the batch
being computed from the same record the hazards are judged against.
\end{proof}

On the checkpointed machine of Section~\ref{sec:checkpoint}, the same
policy also separates that section's pair. The proof confirms that the
delta is not durable-local-measured. It uses information outside the
defeated class instead of refining a policy inside it.

\paragraph{Why ascending coordinates are required.}
P4 does double duty. It is what makes the delta's own batch
internally duplicate-free, and it is what licenses the strongest
reading of the conclusion: each committed obligation delivered
\emph{per instance} exactly once. Without it the count is set-level, by
payload. That distinction matters. An equal-coordinate aliasing
counterexample in the mechanization refutes the per-instance strengthening without ascending coordinates at the deduplication
tier (Section~\ref{sec:memory}). One further property comes free and
one does not: the delta emits in committed source order (a proved
fact), while ordering across a real transport is a
different guarantee, which this paper does not provide.
Section~\ref{sec:wire} models exactly where it is lost.

\paragraph{Normal execution.}
The delta chooses the right batch at recovery. A companion result
proves the running discipline that keeps a live, emitting system
hazard-free between recoveries, as three record-local guards:
\begin{itemize}
\item[G1.] \emph{Fire-time justification:} an ordinary publish fires
  only a payload the committed history holds at fire time.
\item[G2.] \emph{Publish freshness:} it fires only a payload not
  already on the record.
\item[G3.] \emph{Fresh reconcile suffix:} a re-drive's suffix is
  internally distinct and wholly absent from the record.
\end{itemize}
Every run keeping the three guards stays hazard-free. The folklore
version of this rule has two guards. A separate counterexample shows
why each guard is required, and
suffix freshness is provably necessary once a re-drive can meet a
record that already holds part of its window. These guards do not
guarantee completeness. An empty re-drive suffix satisfies them while
silently skipping work. The guards constrain what fires, not what
fails to fire.

\paragraph{Interrupted and repeated recovery.}
Under P1-P4, the small-step variant can deliver a prefix of the
sink-delta batch and stop. Within this recovery window, the frontier
and every field except the ledger remain fixed. Recomputing the delta
returns the undelivered suffix. Partial rounds remain hazard-free.
Delivering the remainder and repairing the store completes recovery.

A separate result covers repeated atomic sink-delta recoveries and
\textsf{Resume} steps. Ordinary execution between them must satisfy
the model's guards. From a hazard-free start,
with ascending coordinates at each re-drive, every completed re-drive
is exactly-once at its own frontier. Neither result ensures eventual
completion. The status-only \textsf{Recover} step can mark a run
\Recovered{} without repairing its store. It is distinct from the
sink-delta re-drive of Theorem~\ref{thm:escape}.

The observation bound also holds in the small-step extension,
including finite adaptive sequences of per-entry choices that cannot
read the sink record. It therefore does not require atomic recovery.
Crashes inside the fenced recovery operation of
Section~\ref{sec:wire} remain outside that model.

\paragraph{Detecting omissions.}
Exactly-once has two conjuncts, and operators observe them asymmetrically.
A duplicate is loud: the customer gets two emails, and an incident is
opened. A frontier-relative omission is silent: a cursor advanced
past an undelivered send leaves a record with no hazard at all, and a
checked trace reaches a post-recovery state whose computed
hazard verdict is clean while a committed operation remains undelivered at the
judged frontier. The opening incident was the loud kind. The quiet
kind passes the same store-level audit, which is why
completeness needs to be checked separately.

\section{Stragglers and the Arrival Fence}\label{sec:wire}

Theorem~\ref{thm:escape} was proved on a machine where a send and its
acceptance coincide. Real deliveries travel. They wait in retry queues
and socket buffers, arrive out of order, or vanish. The channel machine
models that traffic, and on it the sink-delta policy of Section~\ref{sec:door}
stops being enough: a read cannot see what is still in flight.

\paragraph{The channel machine.}
A channel state is $u = \langle t, W, A, \varphi \rangle$: an inner
effect state $t$, the \emph{wire} $W$ (emissions sent and not yet
resolved), the sink's durable \emph{accepted record} $A$, and a
\emph{fence} $\varphi \in \mathbb{N}$, zero until set. An ordinary publish adds its
emission to $W$ and to the inner ledger $E$, which here records sends
rather than acceptances. An arrival resolves one in-flight emission, in any order
(the wire reorders freely), and the sink accepts it into $A$ exactly
when its generation clears the fence, $\varphi \le \gen(x)$, dropping
it otherwise. An in-flight emission may instead be lost, as a
first-class step. The bookkeeping is conservative (every sent emission
is at all times accounted for as in flight, accepted, or dropped), and
a crash moves the run's status and nothing else: the wire, the accepted
record, and the fence all survive it, because a process crash does not
empty network buffers. Hazards and completeness are now judged at the
accepted record ($\unsafe_A(u)$ for a premature or duplicate
\emph{accepted} entry, $\alo_A(f,u)$ for at-least-once read there), and
exactness at $f$ is the conjunction
\[
\begin{aligned}
  \eo(f,u) \;&\longleftrightarrow\;
     \neg\,\unsafe_A(u) \,\wedge\, \alo_A(f,u),\\
  \alo_A(f,u) \;&\longleftrightarrow\;
     \replay(H,\scope,f) \subseteq \pay(A(u)).
\end{aligned}
\]
The unfenced policy re-drive studied below atomically repairs the
store and appends its batch to both $E$ and $A$. The batch is accepted
synchronously without changing the wire or fence. Cursor-based
recovery, using the suffix replay described in
Section~\ref{sec:checkpoint}, instead queues its re-sends on the wire.
Theorem~\ref{thm:wire} concerns the unfenced recovery rule just
described. The later fenced rule is outside its scope because it
changes the fence.
Deployed systems keep a cousin of the generation under another name:
Kafka's idempotent producer carries a broker-checked
\emph{epoch}~\cite{kip98}, and this paper uses that word only for the
cited mechanism.

\paragraph{An old request can arrive after recovery.}
Figure~\ref{fig:straggler} shows an old send that remains in flight
across the crash. Recovery reads the accepted record, sends the missing
payload, and resumes in a new generation. The old request then arrives
and the sink accepts it, producing a duplicate. Every entry is
individually justified. The duplicate comes from the run's own valid
send in the previous generation, not from malformed traffic.

When the old request arrives \emph{before} the read, the read sees it
and the delta is empty. This control shows that the sink read is
correct when it occurs. It becomes stale only when the wire still
contains work that can arrive after the read.

\begin{figure}[t]
\centering
\includegraphics[width=\columnwidth]{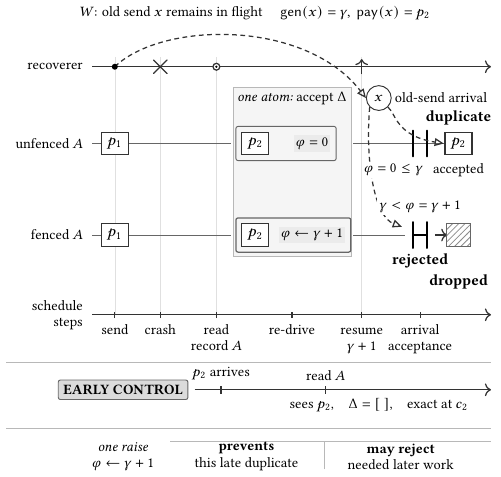}
\Description{A dashed old send spans a crash, a sink read, the
recovery atom, and resume before reaching a circled arrival point. From
there, separate rows show an open unfenced gate accepting a duplicate
payload and a closed fenced gate rejecting the item into a hatched
drop. Lower rails show an arrival-before-read control and the benefit and
limitation of raising the fence.}
\caption{An old send crosses the sink read. Without a fence, its later
arrival duplicates the accepted record. The fenced re-drive atomically
lands the recovery delta and raises $\varphi$, so the arrival is rejected
(Theorems~\ref{thm:wire}-\ref{thm:fence}). The lower rails show the
benign early-arrival control and the rejection of later work by the fence
(Corollary~\ref{cor:price}).}
\label{fig:straggler}
\end{figure}

\begin{definition}[Channel-blind policy]\label{def:blind}
A policy $P$ on the channel machine is \emph{channel-blind} when it
factors through everything but the wire: for some $g$,
\[
  P(u) \;=\; g\bigl(\langle\, t,\; A(u),\; \varphi(u) \,\rangle\bigr)
  \quad\text{for all } u,
\]
where $t$ is the whole inner state, including the store, ledger, and
generation.
\end{definition}

Everything durable is inside the visible triple, the accepted record
and the fence included. The one field outside it is the wire itself.

\begin{theorem}[Wire Bound]\label{thm:wire}
On the channel machine there is a reachable pair $u_1, u_2$ that agree
on the visible triple. One has the crashed process's own send still
in flight, while the other has an empty wire. This pair defeats every
channel-blind policy $P$ under the unfenced re-drive relation: there
are re-drives of both members at $c_2$ whose outcomes satisfy
\[
  \bigl(\exists v.\ u_1' \xrightarrow{\,\mathrm{arrive}\,} v
        \,\wedge\, \unsafe_A(v)\bigr)
  \;\vee\; \neg\,\alo_A(c_2,\, u_2') .
\]
Either the straggler's arrival leaves the first member's accepted
record unsafe, or the second member still lacks a delivery at $c_2$.
\end{theorem}

Because the policy cannot observe the wire, it chooses the same
batch for both states. If the
batch carries the missing payload, the in-flight member's straggler
arrives at the never-set fence and the record duplicates. If it omits
it, the empty-wire state has no pending request that could deliver the
missing event and remains incomplete. The
sink-delta policy of Section~\ref{sec:door} is in the defeated
class ($\sdelta$ reads the accepted record and never the wire), and
the late-arrival example above demonstrates this failure. A variant of the
opening example has the same hazard: a send left in a retry queue at
crash time can produce a duplicate after an otherwise correct sink
read. The hazard is not an artifact of the model.
The mechanization separates the late-arrival example's final record from
everything the instant-delivery machine can reach, so the guarantee of
Section~\ref{sec:door} was not false on its own machine. That machine
could not express this failure.

\paragraph{The arrival fence.}
The \emph{fenced re-drive} at $f$ is one atomic act: heal the store,
compute $\sdelta$ from the accepted record, land the batch in the sent
ledger and synchronously in $A$ (a recovery's own re-sends never
traverse the fenced wire), and set the fence to the successor of the
crash-time generation, $\varphi := \gamma + 1$. From then on the sink
rejects arrivals from generations below the fence. The guarantee has four
conditions corresponding to
P1-P4, read at the accepted record (record
hazard-free, crashed, in-interval, and ascending coordinates, called
Q1-Q4) and, for the continuation results below, the machine's
reachability invariants.

\begin{theorem}[Arrival Fence]\label{thm:fence}
On the channel machine, under Q1-Q4, the fenced re-drive at $f$
exists, and every fenced re-drive $u \rightsquigarrow u'$ leaves the
store mismatch-free at $f$, is exactly-once at $f$ on the accepted
record, and sets the fence to the successor generation:
\[
  \clean(u',f), \qquad \eo(f,u'), \qquad
  \varphi(u') = \gen(u) + 1.
\]
\end{theorem}

The theorem assumes strictly ascending coordinates because
duplicate-freedom of the accepted record depends on them. The theorem
certifies this admission mechanism under the machine's crash-time
stamping discipline. It does
not say a generation fence is the only possible way to neutralize a
wire. A sink that atomically deduplicates on stable event identity at
admission, for instance, is a different mechanism with different
assumptions, unmodeled here.

\paragraph{Fence placement.}
The composite's shape is forced by a discipline the machine inherits:
re-drives stamp at the crash-time generation (the reconcile is
generation-constant, and resume bumps the counter only afterwards), so
an earlier process run's in-flight emission and the recovery's own
re-drive carry the \emph{same} generation, and no fence value separates
them on the wire. Set the fence at resume instead, and the straggler
slips in during the gap. Setting it in a separate step before the
re-drive would cause a recovery whose own sends traveled the wire to
fence out its own work. Synchronous acceptance plus an atomic rise is the one sound
point in this rule family. It is a design argument over the stamping
discipline, while the two refuted
reorderings that \emph{are} theorems live one machine over, in
Section~\ref{sec:claim}. Figure~\ref{fig:fences} places the sound rise
beside its failing neighbors.

The fence's policy is still the channel-blind $\sdelta$. Nothing reads
the wire. That is the central point: more recovery-side reading
does not flip the verdict (that was Theorem~\ref{thm:controlplane}),
while the same generation tag, consulted at \emph{acceptance}, does.
The fence works by changing what the sink admits, not by reading more.

\paragraph{Cost of the arrival fence.}
The following corollary shows that fencing can reject a delayed
delivery needed at a later frontier.

\begin{corollary}[Rejection of later work]\label{cor:price}
There is one schedule over fenced and unfenced twins that differ by one
rule. Recovery occurs at frontier $c_1$ with a straggler for $c_2$ in
flight. On this schedule, the unfenced run's straggler arrives and completes
the record at $c_2$, exactly-once there with the payload counted once,
while the fenced run drops the same straggler: its result remains incomplete at
$c_2$, and exactness holds at the fence's own frontier $c_1$.
\end{corollary}

The fence guarantees exactly-once acceptance at the chosen frontier.
It can still leave later work undelivered by rejecting the old request
that would have delivered it.

\begin{corollary}[Residual-wire stability]\label{cor:wire-stability}
Under Q1-Q4 and the machine's reachability invariants before a fenced
re-drive at $f$, every wire entry left after the composite carries a
generation below the fence, and along any subsequent trace of arrivals
and losses only, the accepted record is unchanged and exactness at $f$
persists.
\end{corollary}

Stale arrivals cannot be admitted by reordering or delay, and losses
cannot alter $A$. The corollary is scoped to arrival/loss-only
continuations on purpose: it says nothing about arbitrary later
application work. Theorem~\ref{thm:fence} applies to a later recovery
only if its conditions hold at that crash.
Between recoveries, an in-flight entry is neither delivered nor lost at
a completed frontier. The model keeps that ambiguity first-class rather
than resolving it by fiat.

\paragraph{Validation controls.}
Two companion results locate the modeling choices. On a twin machine
whose crash wipes the wire, the same schedule ends exactly-once with a
duplicate-free record. Crash-survival of the wire is therefore a
required assumption and a realistic choice because crashes do not empty a
network's buffers or a sidecar's retry queue. And the early-arrival
schedule above calibrates from the other side: arrival before the read
is benign, so nothing weaker than arrival-after-read defeats reading.
The model treats the fenced re-drive as one atomic transition and does
not cover a crash inside it (Section~\ref{sec:mech}).

\section{Concurrent Recoverers and the Claim Fence}\label{sec:claim}
\begin{figure*}[t]
\centering
\includegraphics[width=\textwidth]{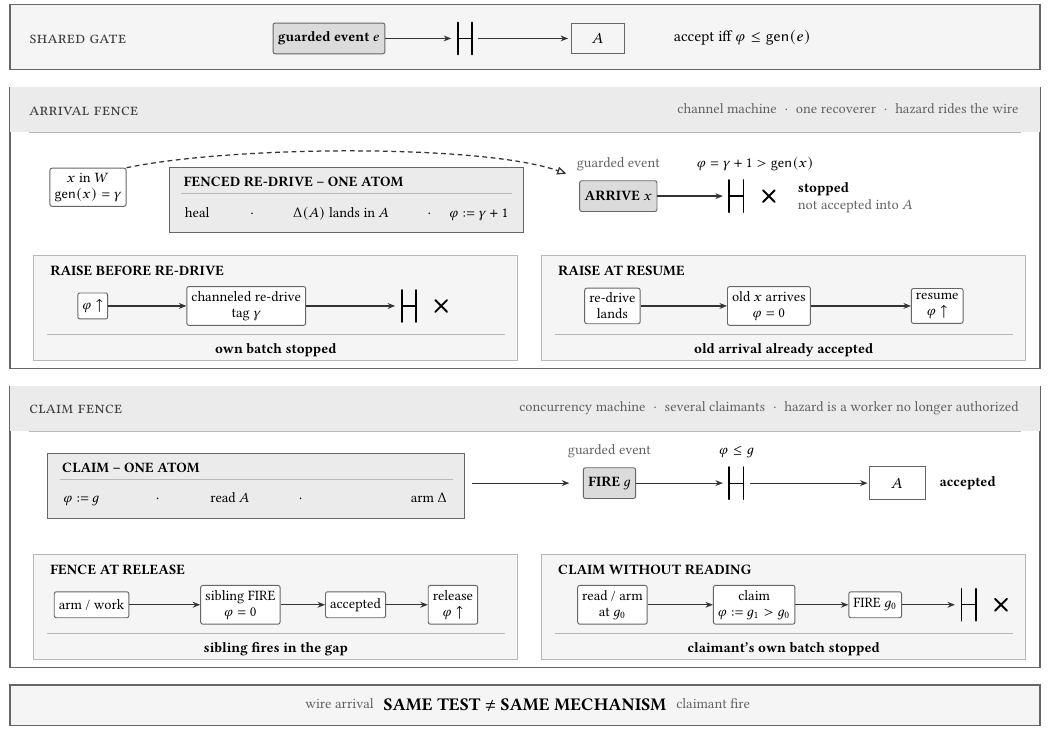}
\Description{A shared-gate key sits above two full-width mechanism
rows. The arrival-fence row shows an old wire item rejected after an
atomic fenced re-drive, with separate cards for raising before re-drive
and at resume. The claim-fence row shows a fire accepted after an
atomic claim, with separate cards for fencing at release and claiming
without reading. A footer says that the same test does not imply the
same mechanism.}
\caption{Arrival and claim fences use the same acceptance test for
different events. The channel machine fences arrivals after an atomic
re-drive. The concurrency machine fences fires after an atomic claim.
The side cards show invalid reorderings. Only safety transfers between
the two regimes.}
\label{fig:fences}
\end{figure*}

A second hazard family comes from orchestration. An orchestrator that
loses a worker's heartbeat can start a second recovery while the first
is still alive: two recoverers, each reading honestly, each re-driving
what it read. This section is about that pair, on a machine built for
it.

\paragraph{The concurrency machine.}
The machine carries
$w = \langle \sigma, J, S, A, \varphi, \Gamma, \mathit{hwm},
\mathit{fl} \rangle$. Its fields are the store, a \emph{journal} $J$ (the record
justification is measured against, whose role sharpens in
Section~\ref{sec:memory}), \emph{sent} and \emph{accepted} ledgers, a
fence, a \emph{generation table} $\Gamma$ tracking each spawned
generation as producing or armed with a batch, the high-water mark
$\mathit{hwm}$ (the newest spawned generation), and a retention
\emph{floor} $\mathit{fl}$. Delivery happens at the \emph{fire point}:
an armed generation $g$ fires its batch, which always lands in $S$ and
lands in $A$ exactly when the fence has not moved past it,
$\varphi \le g$. The sink delta lifts to this machine, computed against
the accepted record over what retention has kept. There is no wire
here by design. The hazards of this machine come from \emph{workers that remain alive
after losing recovery authority}, not from messages in flight, and the two
hazard families do not substitute for each other. Before anything new
begins, the observation bound reappears in a strengthened form. A policy
reading the store, the journal's length, the fence, the whole
generation table, the high-water mark, and the floor (everything short
of the ledgers) is still defeated on a designed pair of this machine.

Three authorities must stay separate on this machine, because
Section~\ref{sec:memory} pulls them apart. The live source history $H$
sits inside $\sigma$. The retention-relative obligation history is
$H_{\mathit{ret}}=\mathrm{drop}(\mathit{fl},H)$. Delta, loss,
at-least-once, and exactness are evaluated against replay from
$H_{\mathit{ret}}$. The journal $J$ is the authority that justifies
emissions, while $A$ detects prematurity and duplication. Exactness on
this machine is therefore \emph{retained-view-relative} unless an additional
condition identifies $J$, $H$, and $H_{\mathit{ret}}$, a relativity this
section will state precisely rather than gloss.

\paragraph{Fence-free one-shot recovery.}
The negative's subject is the natural discipline (each recoverer
reads, arms once, fires once) run twice in parallel without a fence.

\begin{theorem}[Second-Recoverer Bound]\label{thm:second}
On the concurrency machine there is a reachable crashed state $W$
with
\[
  \begin{gathered}
  \varphi(W) = 0, \quad g_1 \neq g_2 \text{ both live}, \quad
  \neg\,\unsafe_A(W),\\[-2pt]
  \sdelta\langle f\rangle(W) = [\,p\,].
  \end{gathered}
\]
The state has four properties: the fence is unset, two generations are
live, the state is hazard-free, and exactly one delivery is missing at
$f$. For every
fence-free pair of one-shot recoverers, however each computes its batch
with the whole state readable, some schedule in the constrained family
ends in a hazard or with the delivery at $f$ still missing. In this family,
both recoverers arm, at least one fires, at least one heals, and every
recoverer that heals also fires.
\end{theorem}

The defeating schedule is four or five steps, and which steps the
adversary needs depends only on where the missing payload sits in the two
armed batches. If both batches carry it, let one recoverer heal and
both fire. The accepted record then holds the payload twice. If neither
carries it, the same steps leave the delivery incomplete. If exactly one
carries it, let that recoverer die silently after arming (the family
permits one silent death) while the other heals and fires its
payload-free batch. The store is repaired, the run reads recovered, and
the delivery is gone.
Each recoverer fixed its batch before the schedule chose the cell. The
family's constraints exclude trivial defeating schedules. They cannot
kill both recoverers and call the loss a defeat, and a recoverer that heals
cannot then vanish (every healer fires). As everywhere in this paper,
the defeat is existential over schedules and says nothing about
disciplines outside the fence-free one-shot class. The positive result
below uses such a discipline. The sink-delta policy of
Section~\ref{sec:door} also fails on this axis. When run by two claimless
recoverers, the sink delta double-fires. Two honest workers, each
individually correct by Section~\ref{sec:door}'s rule, are jointly
wrong: each read was true when taken and stale by the time it was
used.

\paragraph{What ordering changes.}
Suppose each recoverer must prepare its batch, repair the store, and
then send, with no step skipped. The model permits only one store
repair per crash, so these ordered schedules contain at most one
completed recovery. At the constructed state, a completed sink-delta
recovery is exactly-once. The duplicate in
Theorem~\ref{thm:second} therefore requires a recoverer to send without
completing its store repair.

Ordering does not ensure that a batch contains the missing work.
Fully ordered schedules exist for every pair of batch choices. If one
batch omits the missing payload, there is an ordered schedule in
which that worker completes while the other stops after preparing
its batch. The delivery remains missing. The original schedule family
also includes these ordered schedules.

\paragraph{The claim fence.}
An atomic \emph{claim} by generation $g$ requires $\varphi \le g$, so
claims never lower the fence. The claim sets the fence to $g$ and arms
the batch computed from the accepted record read in that step. The later
fire is accepted only if $\varphi \le g$, meaning it
\emph{completed at the fence}.

\begin{lemma}[Claim-fence safety]\label{lem:claim}
Every claim-disciplined trace from the machine's initial state is
hazard-free.
\end{lemma}

The safety proof requires the following restrictions in addition to the
machine's usual state and history checks.
\begin{itemize}
\item Ordinary delivery sends only committed events that are absent from
  the acceptance record. It pauses while any prepared recovery batch
  remains authorized by the fence.
\item A claim prepares a duplicate-free batch from the sink's acceptance
  record. The alternative cursor-recovery sequence must use a
  duplicate-free suffix containing only events absent from that record.
  It cannot start while another prepared batch remains authorized, and
  its prepare, repair, and send steps must run together without
  intervening actions.
\item Snapshot effects must match the state reconstructed from the journal
  prefix. These effects are used by the DBLog instance described in
  Section~\ref{sec:mech}.
\end{itemize}
The permitted executions exclude truncating crashes, independent fence
raises, and batch preparation outside a claim or the cursor-recovery
sequence above. Retention, lost publishes, release, and sending an already
prepared batch remain permitted.

Under these rules, the fence never exceeds the newest spawned generation.
A worker can still attempt to send after a later claim removes its
authority. The sink therefore checks the worker's generation when the
batch is sent.

\begin{theorem}[Completed-Claim Exactness]\label{thm:claim}
On the concurrency machine, under
\begin{itemize}
\item[C1.] \emph{Initialized discipline:} the trace starts at the
  initial state and follows the claim discipline.
\item[C2.] \emph{Claimed:} generation $g$ claims at frontier $f$ and
  later fires, with no intervening claim or fire of $g$.
\item[C3.] \emph{Window-fresh:} no source commit at or below $f$ lands
  between the claim and the fire.
\item[C4.] \emph{Completed at the fence:} $\varphi \le g$ when the
  fire lands, so its batch is accepted.
\end{itemize}
the post-fire state is exactly-once at $f$, judged against the
retained obligation history $H_{\mathit{ret}}$. Per-instance exactness
additionally requires ascending coordinates. Without them the count is
set-level, by payload.
\end{theorem}

Window freshness is required. A commit at or below
$f$ landing inside the claim window grows the set of required deliveries after the batch
was fixed. One modeled idealization must be stated here: the fire
transition executes its armed batch atomically. A recoverer crashing after a strict prefix of a
multi-item batch is outside the modeled transition. Per-entry firing is
future work, and Section~\ref{sec:mech} carries the absence.

\paragraph{Retention-relative exactness.}
The theorem's conclusion is graded against what retention kept, and
the difference matters. The machine's retention step has no guard, so
a disciplined run can raise the floor over
an obligation that was committed and never delivered, dropping it from
the graded history without disturbing the theorem's verdict. The
mechanization includes a concrete control trace that is exactly-once
relative to retained history while, judged against the journal,
a committed obligation is lost. That run's accepted record is empty. The
stronger, journal-graded statement is then closed under either of two
conditions, and the boundary between them is itself a theorem. On the
\emph{permanent-source} fragment (journal equal to the committed
history, floor zero), the two grades provably coincide, and
Theorem~\ref{thm:claim} yields journal-grade exactness outright. Off
that fragment, the bridge is \emph{floor-prefix coverage}: everything
retention dropped below the floor was already accepted.

\begin{corollary}[Journal-grade exactness]\label{cor:journal}
With the journal equal to the committed source history,
\[
  \eo^{J}(f,w) \;\longleftrightarrow\;
  \eo^{\mathit{ret}}(f,w) \,\wedge\, \mathrm{covered}_{\mathit{fl}}(f,w),
\]
where $\eo^{J}$ and $\eo^{\mathit{ret}}$ grade exactness against the
journal and the retained history respectively, and
$\mathrm{covered}_{\mathit{fl}}(f,w)$ means that every scoped source
event at or below $f$ that retention has excluded from replay already
appears in the acceptance record. In
particular, a disciplined run whose retention steps only ever drop
already-accepted obligations delivers journal-grade exactness through
Theorem~\ref{thm:claim}.
\end{corollary}

The equivalence is the precise form of a folklore sentence (``replay
what retention kept, and make sure retention only expires what was
delivered''), and the guard-free counterexample above is why the second
clause cannot be dropped.

\paragraph{Two invalid orderings.}
The discipline depends on this order. Two concrete counterexample
traces refute the tempting variants
(Figure~\ref{fig:fences}, lower row). \emph{Fencing at release}
(doing the work first and setting the fence at hand-off) admits a batch from an
earlier worker before the fence is raised. \emph{Claiming without
reading} (setting the fence over a batch computed before the
claim) fences out the claimant's own re-drive. Claim first, read under
the claim, fire under it.

\paragraph{Interface between the two fences.}
This paper now has two fences, and no theorem merges them. The arrival
fence guards a \emph{wire}: one recoverer, arrivals from its own past.
The claim fence guards a \emph{fire point}: several recoverers, including workers
that can still send after losing recovery authority. They live on different machines. What is
proved between the regimes is one interface with two obligations:
\begin{itemize}
\item[(T)] \emph{Transport delivers a sub-multiset:} what arrives is a
  sub-multiset of what was sent (reorder, delay, and drop are
  allowed, but fabrication, duplication, and alteration are not).
\item[(J)] \emph{Justification is monotone:} a record that justifies
  an emission still justifies it after any extension.
\end{itemize}
The claim discipline's safety argument is proved once against these
obligations alone, then instantiated twice: by a genuine pullback from
this section's fire point, and by a sibling abstract wire grammar that
permits loss, reordering, and arrivals from earlier generations. That
sibling is \emph{not} Section~\ref{sec:wire}'s channel machine and does
not re-derive its theorems. The asymmetry is the interface's shape:
safety crosses it, exactness does not, and nothing else crosses at
all. Note also what obligation (T) excludes: transport that
\emph{duplicates} a staged item. A generation fence rejects requests from
generations below its current threshold. It does not deduplicate a same-generation copy the
transport minted on its own. Deployments running over at-least-once
transport need stable identity and an admission or deduplication
contract in addition to anything this section proves, and
Section~\ref{sec:memory} is about how long such contracts hold.

\section{Evidence Lifetime}\label{sec:memory}

Systems often combine at-least-once delivery with
deduplication~\cite{treat,ddia}. Theorem~\ref{thm:dedup} gives a
coverage identity for an ideal sink view with permanent memory. The
remaining results examine bounded memory and loss of source history.

\paragraph{Permanent deduplication.}
Model the idempotent consumer as the \emph{absorbing, permanent-memory
deduplicated view} of the sink's record, in which repeated payloads are
absorbed,
$\dedupv(t) = \mathrm{remdups}(\pay(E(t)))$, remembering every payload
it has ever seen. For a list $L$, $\#_p L$ is the multiplicity of
payload $p$.

\begin{theorem}[Deduplicated-View Coverage Identity]\label{thm:dedup}
On the effect machine with the absorbing, permanent-memory view: every
payload in the set of scoped replay obligations at $f$ is represented
exactly once in the deduplicated view if and only if delivery is
at-least-once at $f$:
\[
  \bigl(\forall p \in \replay(H,\scope,f).\;
        \#_p\,\dedupv(t) = 1\bigr)
  \;\longleftrightarrow\; \alo(f,t).
\]
\end{theorem}

The view is duplicate-free by construction. The theorem establishes
whether every required payload is present. It does not rule out
premature effects on the underlying record or verify a consumer's
implementation of checking for a prior effect, applying the effect,
and recording its completion.

The per-instance corollary additionally requires strictly increasing
source coordinates. An equal-coordinate example refutes that
stronger conclusion without the condition. The following results
examine the assumption of permanent memory.

\paragraph{Bounded deduplication.}
Real deduplication memory is a window, and the window here is an
\emph{entry-count} horizon, measured in record positions, not
seconds. A duration maps to it only through a delivery-rate argument.
A legal-run counterexample shows how the guarantee fails. A windowed view satisfies the same
equivalence at a reachable state, and then one more crash and one more
honest re-drive push the original delivery out of the window, so the
sink processes the re-driven copy as new. The windowed count reads
two while permanent memory still reads one. For the constructed run
the threshold is sharp between window sizes one and two, and the
window-sufficiency condition that protects a run is proved sufficient
and \emph{not} necessary. The operational shape of the hazard is
familiar: Stripe's idempotency keys become eligible for removal after
at least 24 hours~\cite{stripe}, and Kafka's producer-ID expiration setting defaults to one day,
although ongoing transactions and topic retention affect when state
expires~\cite{kafka-exp}. An
outage that outlives the deduplication horizon reproduces the
incident in its quiet form: recovery is honest, the keys are gone, and
the re-drive lands as new work.

\paragraph{Deduplicating by key.}
Deployed systems often match events by keys such as an LSN, offset,
or idempotency key~\cite{stripe}. The keyed policy compares source
coordinates instead of the full payload $(c,e)$. It computes the
same missing batch when source coordinates increase strictly and
every accepted entry matches a committed source record.

Each condition has a separate counterexample. Coordinate matching can remove entries from the batch selected by
full payload matching, but cannot add entries. The counterexamples
show how this can suppress a required replay and leave committed work
undelivered. Comparing only application contents is
also insufficient. Distinct committed operations can have identical
contents, and their source coordinates distinguish them. Derive the
key once, at commit time, from that stable coordinate.

\paragraph{Compaction that keeps presence.}
On the sink's own record the safe boundary is \emph{presence}.
Rewrites that preserve a record's presence (compaction, tombstoning
that keeps the key's entry) do not affect the windowed policy,
which needs no more than window-local presence. Presence-destroying
retention fabricates duplicates at a reachable state. The windowed
policy never loses work, so its only failure is duplication. This
bounds what the policy analyzed here needs to read. That
no policy could do better with bounded memory is not claimed
(Section~\ref{sec:mech}).

\paragraph{Compensation instead of deduplication.}
Sinks with reversing entries permit another approach: compensate,
then redo. The theory analyzes three selected compensation
observations (matched cancellation, signed net, and truncation), and
the verdicts split. Under the matched-cancellation view, order-aware
compensate-then-redo normalizes both outcomes of the designed pair.
The same batch succeeds on both members because the judged observation
has changed. Under the signed-net observation, the construction again forces
recovery to choose the same batch for states that require different
actions. Truncated views do not distinguish the pair either. Throughout, the underlying record
stays append-only. Reversal operations never delete sent events, but instead append matching
reversing emissions in sequence. The matched-cancellation result is proved for
the abstract re-drive
relation. It is not an execution of this machine's suffix-only reconcile. Realizing it
requires a compensation-capable transition, or a refinement proof for
an implementation that has one.

\paragraph{Source truncation and journal evidence.}
The second shrinkage is upstream, and it is the harsher one. On the
concurrency machine, retention raises the floor over a history that
still exists. A \emph{truncating crash} instead cuts the live source
history itself, while the journal $J$ remains the specification of
which deliveries were required. What recovery can read afterwards is the
\emph{retained view}
\[
  \view(w) = \langle\, \sigma,\ |J|,\ \varphi,\ \Gamma,\
  \mathit{hwm},\ \mathit{fl},\ S,\ A \,\rangle ,
\]
everything a policy could ask for, except that the journal enters
only as its \emph{length}. Now take two runs that each commit one
event, differing only in which, and let a truncating crash cut both
live histories back to empty. The surviving stores are literally
equal, and so is every other slot of the view: same fence, same
generation table, same high-water mark, same floor, same two ledgers,
and journals of the same length. The journal contents differ, hidden
from the policy, and with them the sets of required deliveries. One
state still has an undelivered event at $f$. The other requires no
delivery at that frontier.

\begin{theorem}[Truncation Dilemma]\label{thm:trunc}
On the concurrency machine there are reachable states $W_1, W_2$ with
equal retained views over differing journals. State $W_1$ has an undelivered event at $f$
according to its journal, while $W_2$ requires no delivery there. Every
policy computed from the retained view is defeated: there are
re-drives of both at $f$ under which $W_2$'s outcome is a hazard (a
fabricated delivery where none was required) or $W_1$ remains incomplete at $f$
(abandonment).
\end{theorem}

Whatever rule the policy runs, one state makes it fabricate work or
the other makes it abandon work. Two
follow-ups make the theorem usable.
Reading the \emph{live} source instead of the journal does not help
and actively harms: the machine admits phantoms (deliveries whose
justifying record is gone), and judging against the live history makes
them retroactively premature, which is the machine-checked reason
justification measures against the journal in the first place. And
the practitioner's version of the boundary is about \emph{content}
versus \emph{extent}: retention that hides old content while an
authoritative journal can still reconstruct it is recoverable. A
crash that leaves recovery with only length-like metadata is what
the theorem describes. Keep the recovery window's journal content
reconstructible, not just its size.

\section{Store-Tier Characterization of Log-Derived State}\label{sec:store}

CDC and the transactional outbox are the standard remedy for dual
writes, and they are widely deployed: Debezium alone ships log
connectors for most major databases~\cite{debezium}. Everything so far
has treated ``derive downstream state from the committed log'' as
sound advice whose delivery stage still needs a theory. This section
proves the advice's own half exactly, and locates the boundary between
the two halves.

\paragraph{Store-tier characterization.}
At the store tier, the model is Section~\ref{sec:model}'s substrate
with implementations over it. An implementation supplies an initial
state, a step relation, and a projection to core states. Three
properties name what the theorem asks of it:
\begin{itemize}
\item[R1.] \emph{Crash-closed:} every reachable running state can crash at every
  in-interval frontier.
\item[R2.] \emph{Refining:} its steps project to the substrate's.
\item[R3.] \emph{Running-start:} its initial state is \Running.
\end{itemize}
For a materialized map $h$, $h\rst\scope$ is restriction to scoped
keys. An \emph{observable mismatch at $(c,k)$} means the state is
$\Crashed\,c$ with $c\leq\fin$ and the committed and downstream
materializations disagree on scoped key $k$ at $c$.

\begin{theorem}[Faithful-Image Equivalence]\label{thm:image}
For a crash-closed, refining, running-start implementation $I$ of the
store model: $I$ has no reachable observable scoped mismatch at a
crash frontier if and only if $I$ is running-image-faithful at every
reachable \Running{} state and every in-interval frontier $f$,
\[
  \Src(b,D,f)\rst\scope \;=\; \Src(b,H,f)\rst\scope .
\]
\end{theorem}

Unlike the recovery bounds, which are constructions, this equivalence
quantifies over \emph{all} implementations of its tier: the one
result in the paper that is universal over implementations, as promised in
Section~\ref{sec:model}. It characterizes the equality that ``derive
from the log'' must establish: under R1-R3, keeping the running
downstream image faithful to the committed log \emph{is} store-tier
crash safety, in both directions. The two directions are useful in
different places. Faithfulness rules out crash-frontier mismatches
by a structural step: a crash changes no history, so a mismatch frozen
at a crash was already present at a running predecessor, where
faithfulness forbids it. That step needs R3, and R3 is necessary. An
implementation that begins already crashed has no running states at
all, is faithful vacuously, and is unsafe anyway. That class provides
the counterexample when R3 is omitted. The converse is
where R1 works: if some reachable running state held a scoped mismatch
at an in-interval frontier, crash-closure supplies a crash \emph{at
that frontier}, making the mismatch observable. If R1 is weakened
to crashes at only some frontiers, the argument has no crash to take.

\paragraph{Relay example.}
The two-event instance shows the biconditional at work. A downstream
derived from the log in committed order has image $7$ at $c_1$ and $9$
at $c_2$, matching the source at both frontiers. No crash of that run
freezes a mismatch. Now let the relay acknowledge the update at $c_2$
but defer the write. Its delivered history still stops at the insert,
so at frontier $c_2$ the downstream image reads $7$ against the
committed $9$: a scoped mismatch at a running frontier, and by R1 a
crash is available right there to expose it. The practitioner's
version is short: at this tier the crash adds no window of its
own (what a crash can expose is exactly the gap between the
committed log and the derived image at a running frontier), so derive
from the log, in the log's order, and close that gap.

A lagging relay also shows what the theorem does \emph{not} provide.
Between the acknowledgement and the deferred write, the relay is not
running-image-faithful, and no theorem here says a lagging pipeline
``catches up.'' Liveness is outside this model. The constructive
result is frontier-scoped. At every
\emph{occurrence-complete} frontier (one whose scoped obligations
have all been delivered), the image is mismatch-free. CDC and outbox
designs should therefore be read as machinery for establishing
complete-frontier image faithfulness at the frontiers they advertise,
not as global faithfulness at every instant. An advertised frontier is
a claim. The theorem says what that claim must mean.

\paragraph{Limit of the store tier.}
The designed pair of Theorem~\ref{thm:obsbound} agrees on the whole
store, so store-tier faithfulness cannot arbitrate a single
re-delivery decision. The acceptance boundary begins where this
biconditional stops. The tier scales to multiple derived targets
without new theory. The proof shows that fan-out diverges
exactly when some single target does, and a target in the formal
log-derived class (which requires image equality at every modeled
frontier) is excluded from the divergence class. That is a statement
about the formal class, not a blessing of every architecturally
log-backed but lagging deployment. The per-target reading is the right
one, and no cross-target atomicity is implied.

\paragraph{Derivation relocates the delivery boundary.}
A log-derived pipeline still has to deliver. The relay that tails $H$
and writes the sink performs, per event, two durable effects (the
delivery and its own progress record) with no shared commit between
them: Section~\ref{sec:model}'s effect machine under an operational
name. This relocation claim is an argument assembled from proved
parts (Theorem~\ref{thm:image} below the boundary,
Theorems~\ref{thm:obsbound}-\ref{thm:second} above it), and we state
it at that level rather than as a theorem of its own. In the
incident's terms: the relay tailing committed orders was already this
stage. Derivation had removed the handler's dual write and left the
delivery loop's own, which is where the crash found it. The replayable log lets
the relay retry failed deliveries and catch up after an outage. The sink
must still provide the acceptance evidence and controls recovery needs. Nor is
``derive from the log'' a single
design point: the store tier's replay interface has two genuinely
distinct certified inhabitants, a refresh-free outbox segment and the
snapshot-replay instance of Section~\ref{sec:mech}, and mechanized
witnesses establish their distinctness.

\paragraph{Two roles for an outbox.}
Table~\ref{tab:subst}'s role split matters here. An outbox
table read by the relay as its source \emph{is} the committed log $H$.
It is not thereby an acceptance record. The same physical table
realizes $A$ only if it holds sink-owned acceptance state updated
atomically with acceptance. Otherwise $A$ is a separate durable
record on the sink's side. Sentences about ``the outbox pattern''
that do not name the role equivocate between the two sides of the
acceptance boundary: $H$ is what Theorem~\ref{thm:image} governs,
while $A$ is what Theorems~\ref{thm:escape} and~\ref{thm:fence}
consume.

\section{The Mechanization}\label{sec:mech}

This section explains how the recovery models relate and what the
Isabelle proofs establish.

\paragraph{Relationship between the machines.}
Write $\iota$ for the embedding of effect-machine states into the
concurrency machine and $\Pi$ for the projection back.

\begin{theorem}[Shared-Core Agreement]\label{thm:agree}
(i) $\iota$ embeds every effect-machine trace into the concurrency
machine's solo fragment. Projection $\Pi$ maps every solo-fragment trace
from a solo-invariant state back while preserving the invariant, and
$\Pi \circ \iota$ is the identity. (ii) Under the alignment
conditions that everything sent is accepted, the journal equals the
committed history, and all payloads are log entries, the hazard verdicts $\prem$,
$\dupl$, and $\unsafe$ agree through $\Pi$. (iii) With the ledgers equal
and the retention floor at zero, the sink-delta batch and the loss and
at-least-once verdicts commute with $\Pi$.
\end{theorem}

So the small machine of Sections~\ref{sec:wall}-\ref{sec:door}
really is a fragment of the large one of
Sections~\ref{sec:claim}-\ref{sec:memory}: same verdicts, same
deltas, under conditions that align the larger machine with the smaller
one. The fragment excludes several behaviors. The solo
grammar admits ordinary lifted steps, fused recovery, and successive
generation changes, and excludes free arms, free fence raises,
snapshots, lost publishes, truncating crashes, and retention. Of the
shared-core laws, the three-guard fire discipline of
Section~\ref{sec:door} transports through $\Pi$. The rest were
re-proved on the larger machine. Two non-transfers are structural
rather than accidental: the effect machine's absorbing recovery class
does not lift (release and resume break absorption, and the successor
law is open) and the wire is omitted from the concurrency model. The concurrency
machine has none, so Section~\ref{sec:wire}'s results neither follow
from it nor imply anything on it, and the regimes touch only at
Section~\ref{sec:claim}'s interface. Figure~\ref{fig:models} draws
the resulting family.

\begin{figure}[t]
\centering
\includegraphics[width=\columnwidth]{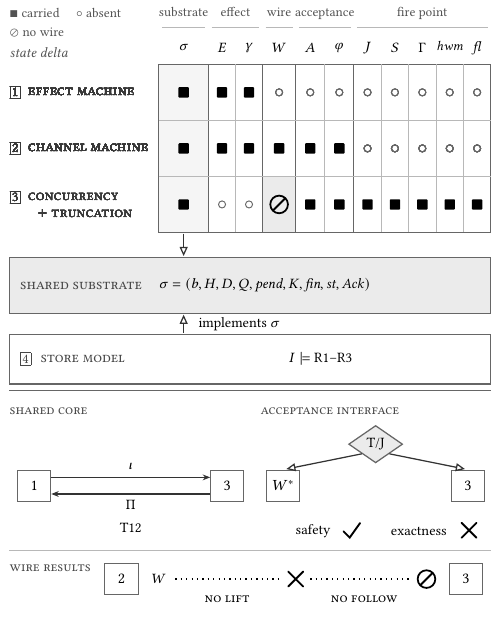}
\Description{A matrix shows the fields carried by three related state
machines over one shared substrate, with a separate store
implementation tier. Machine 3 has a prominent missing-wire cell.
Below, small diagrams show the shared-core projection, the safety-only
interface between the fire pullback and a sibling wire model, and the
blocked channel-wire relation.}
\caption{Fields and relations of the three recovery machines and the
store tier. The checkpointed model extends the effect machine with a
durable cursor and is not shown separately. Filled squares mark
carried fields, and the slashed circle
marks the absent wire. The lower diagrams show the shared-core
embedding, the safety-only interface, and the blocked wire relation.}
\label{fig:models}
\end{figure}

\paragraph{Recovery cycles and the core model.}
\textsf{Resume} changes \Recovered{} to \Running{} and increments
the generation. The histories and emissions ledger stay unchanged.
The core model has no matching transition, so the proof divides
executions of the cyclic machine at \textsf{Resume}.

Each resulting segment corresponds to an execution of the core model
with relay repair added. The store-level theorem applies to sequences
of label steps from the segment's starting state, subject to the stated
checks. It excludes relay repair and further \textsf{Resume} steps.
A checked example preserves a store mismatch across \textsf{Resume}
and repairs it after a later crash.

\paragraph{Prover and trust base.}
Every numbered result in this paper is a machine-checked theorem in
Isabelle/HOL~\cite{isabelle}, checked in a mode where an unfinished
proof fails the build. The trusted base consists of stock Isabelle2025-2
and, beyond the standard prelude, exactly two standard-library
theories (finite multisets and subsequence orderings). There are no
global axioms and no proof oracles. Theorems state their assumptions explicitly. The one
named-assumption context (the acceptance interface) has its two
obligations discharged on both of its instances.
The full development rebuilds from cold in a few minutes on a laptop.

\paragraph{Validation controls.}
The formalization includes four validation controls. Closed witnesses
show that the negative and positive theorem classes are inhabited.
Twin runs test the modeling choices used by the straggler result. A
small-step model confirms that the observation bound does not depend
on atomic recovery. A counterexample found during development required
an additional condition for keyed deduplication. Scratch gate sessions
also checked the main proof terms for unproved dependencies. A planted
failure confirmed that this check detects an unfinished proof. These
scratch sessions are not part of the released development.

\paragraph{Executables.}
Certified executables accompany the theory, chief among them an
exported, theorem-correct store/hazard decider and the core store
machine's trace validator. Both are point evaluators, and their limits
must be stated: a clean hazard verdict is not
exactly-once (completeness at the frontier is the second conjunct,
and no hazard audit sees it), and the validator certifies label
structure, not history wellformedness. In the other direction, the
decider strengthens the observation bound from below. No function of the
durable core alone computes even the hazard verdict.

\paragraph{Model boundaries.}
The model has six limitations.
\begin{itemize}
\item \emph{A crash inside the fenced re-drive.} The composite of
  Theorem~\ref{thm:fence} is one atom. Nothing covers a crash between
  its heal, its landing, and its fence rise. An implementation must
  realize the atom transactionally or supply its own recovery
  argument.
\item \emph{Two-step publish.} A send and the sink's readable record
  coincide (and a fenced recovery's own re-sends land synchronously).
  A sink that acknowledges first and persists later is outside the
  model.
\item \emph{Per-entry claim fire.} The claim machine's fire executes
  its armed batch all-or-nothing. A recoverer crashing partway
  through a multi-item batch is outside the modeled transition.
\item \emph{Truncation on the wire machine.} Fabricate-or-abandon is
  proved at the fire-point machine. The wire-tier lift is open, and
  no wire-tier truncation result exists here.
\item \emph{Bounded-memory converse.} Section~\ref{sec:memory} bounds
  what the policy analyzed here needs to read. It does not prove that every
  bounded-memory policy requires the same information.
\item \emph{Successor recovery-class law.} The effect machine's
  recovery-cycle structure stops at the embedding of
  Theorem~\ref{thm:agree}. Its concurrency-machine successor is open.
\end{itemize}

\paragraph{Scope of the model.}
The absences above concern mechanics. One more concerns shape. Every
result in
this paper concerns a single one-way pair: one authority holding a
durable, ordered, per-operation record that defines the required deliveries, one
endpoint whose durable acceptance defines the done work, one delivery
direction between them. That shape covers the log-derived pipelines
this paper grew out of, and it excludes real configurations by design:
obligations running both ways, several sources feeding one sink, or
one source feeding several sinks bound together rather than held each
to its own pair, chains judged end to end rather than hop by hop,
payloads transformed in flight rather than delivered verbatim, and
pairs in which even the source keeps no committed record of what it
wrote (nothing left to identify the required deliveries, so the model cannot pose
the question). The shape is not an arbitrary cut: its capabilities are
the same ones every result turns on. Each positive theorem states the capability it requires, and each bound
is proved against a named class of policies that lacks it. The generalization this boundary map points toward is a theory that
treats both endpoints as capability-graded nodes of one kind and
recovers these results as the one-way, record-bearing instance.

\paragraph{DBLog instance.}
A pipeline in the style of DBLog mixes an initial snapshot or
backfill with live log records, and the theory covers the mixture at
the safety level. Justification is append-stable across the machine,
mixed-discipline runs stay hazard-free, and the log-only fragment
recovers the delivery hazards exactly as stated. Two boundaries are
explicit. The snapshot modality provides safety guarantees without exactness
laws. And the watermark algorithm itself is not re-proved here: the
instance consumes our prior snapshot-replay development by
citation~\cite{own-arxiv}, through a proved bridge, and that
development's own disclosed assumptions remain its own. One fact
about process belongs here: at the time of writing, no
independent expert has cold-read the formalization. The machine check
and the controls above are the evidence on offer.

\paragraph{Availability.}
The full development is publicly archived as a permanent
deposit~\cite{dual-write-recovery-artifact}, with build instructions
and an index mapping every numbered statement in this paper to its
mechanized counterpart.

\section{Related Work}\label{sec:related}

\paragraph{Operational guidance.}
The operational rule has been public for a decade, in places
peer-reviewed: Kleppmann's dual-writes argument and its book
form~\cite{kleppmann2015dualwrites,ddia}, the log-centric architecture
literature~\cite{kleppmann2019olep}, and Treat's essay on
exactly-once delivery~\cite{treat}. This paper formalizes the recovery
decision that remains after
log-based derivation. Its theorems identify the sink acceptance evidence
needed to decide what to resend and state the conditions for correct
recovery.

\paragraph{Formal studies.}
Closest to our subject are bounded model checks of outbox and
log-delivery designs:
Masternak and Pobiega model-check an exactly-once outbox in
TLA$^{+}$~\cite{masternak}, and LogPlayer model-checks in-order
exactly-once delivery to backend shards~\cite{logplayer}. Both predate
this work. The boundary between us is method, and it matters here: a
bounded check certifies the states it enumerates, while the observation bound of
Sections~\ref{sec:wall}-\ref{sec:checkpoint} is a higher-order claim
over \emph{arbitrary policy functions}, which is the kind of
statement finite enumeration cannot reach. Our result covers arbitrary recovery policies restricted to the modeled
source-side state. We know of no earlier machine-checked proof with this
scope.

\paragraph{Impossibility results.}
The classical results bound our scope rather than supply our
theorems. FLP~\cite{flp-jacm}, in mechanized form
too~\cite{flp-afp}, concerns asynchronous consensus liveness.
CALM~\cite{hellerstein2020calm} characterizes coordination-freeness.
CAP~\cite{gilbert2002cap} trades availability under partition.
Our model concerns a different question: whether recovery can determine
what an independent sink accepted. Mechanized impossibility work such as
Pactole~\cite{auger2013pactole} provides a methodological precedent for
proving limits on classes of protocols.

\paragraph{Atomic commitment.}
The classical way to remove a separable completion window is a shared
commit~\cite{gray78}, and atomic commitment across autonomous stores
has its own theory and impossibility
results~\cite{mullen1992impossibility,guerraoui1995nbac}. Our failure
class instead has the source committed, delivery incomplete,
no coordinator and no votes. Where a real shared commit spans both
effects, the window this paper studies never opens. Where the sink
cannot join one (the usual reason the pattern exists), these
theorems apply. Nothing here bears on commit-protocol correctness.

\paragraph{Verified data systems.}
Mechanized systems work addresses several adjacent problems. VerIso~\cite{veriso} and mechanized
serializability~\cite{afp-serial} treat isolation inside a database,
with no crash-fed second store. GoJournal and
Perennial~\cite{gojournal,perennial} verify one store's crash
safety, with no independently written copy. Disel~\cite{disel} supports compositional verification of
distributed protocols and their clients, including a two-phase-commit
case study. These studies address isolation, crash safety, or protocol
composition. Our focus is the recovery decision at an independently
accepting sink.

\paragraph{Scope of deployed guarantees.}
The model gives deployed exactly-once machinery a vocabulary for
where its guarantees stop. Kafka's idempotent producer and
transactions~\cite{kip98} deduplicate broker-side by producer id,
epoch, and sequence (the epoch is a deployed cousin of our
generation, consulted at acceptance). The producer-ID expiration
setting defaults to one day. IDs do not expire during an associated
ongoing transaction, and topic retention can cause earlier
expiry~\cite{kafka-exp}. These retention limits illustrate the finite
evidence lifetime studied in Section~\ref{sec:memory}.
Flink's two-phase-commit sink~\cite{flink-docs,flink-blog} places
transactional acceptance at the sink. Kafka Streams' exactly-once guarantee
covers processing from source
topics to output topics and state stores~\cite{streams-eos}. That is
a scope boundary, not a defect, and outside it sits the boundary this
paper studies. No broker, connector, or runtime is verified here. The
map says where each guarantee stops, in the model's terms.

\paragraph{Prior use of fencing.}
Fencing by monotonic token is established engineering: Kleppmann's
fencing-token argument~\cite{klepp-lock}, Kafka's producer
epoch~\cite{kip98}, and mechanized lease reasoning in
Grove~\cite{grove}. We claim the theorems, not the mechanism: the
defeat of every channel-blind policy on a designed pair, the conditional
arrival-fence result with its cost stated as a corollary, and the
claim-fence guarantee, to our knowledge the first mechanized
separation of that information bound from a fenced acceptance result.
The corpus models no clocks and no leases, so it neither supports nor
refutes lock-service designs.

\paragraph{Connection to DBLog.}
The worked instance uses the snapshot-replay model from our prior
study of DBLog~\cite{own-arxiv}. Section~\ref{sec:mech} uses that model
as a substrate. It is not prior art for the recovery claims made here.

\section{A Practitioner's Guide to Dual Writes and CDC Recovery}\label{sec:practitioner}

This section translates the formal results into an operational guide for
backend software engineers, systems architects, and reliability practitioners.
Statements here are direct readings of numbered theorems or engineering
guidance based on them.
Throughout, the \emph{frontier} is the position in the committed source
log up to which delivery is judged, and a relay's \emph{generation} is
its restart counter, which identifies the current run.
We address two production architectures:
(1)~direct dual writes in application request handlers, and
(2)~transactional outbox and CDC pipelines with asynchronous relays.

\paragraph{The dual-write dilemma across one process and a CDC relay.}
Dual writes arise whenever an application must record related facts across
two systems that do not participate in a shared transaction (like a two-phase commit). 
In distributed architectures, this problem manifests in two primary patterns.

\emph{Pattern 1: Direct dual writes in application handlers.}
An application service receives a user request (like an order checkout)
and attempts to update its primary database while publishing an event to a search index, or partner API.
If the database transaction commits first and the application crashes before
invoking the external API, the database holds a committed record whose external
side effect was never triggered.
If the external API call succeeds first and the database transaction
subsequently rolls back, an external side effect was released for a transaction
that never existed.
Network calls across boundaries are also three-valued (success,
failure, or timeout): on timeout, the caller cannot determine whether the
packet was lost before arrival or only the acknowledgment was lost, so
retrying can duplicate side effects. Distributed transactions
(XA/2PC)~\cite{gray78} are widely avoided in distributed architectures due to
blocking latency, coordinator failure risks, and the lack of 2PC
support across cloud services and external APIs.

\emph{Pattern 2: The transactional outbox and CDC relay.}
To eliminate direct dual writes in application code, the standard industry
pattern is the transactional outbox combined with change data
capture (CDC)~\cite{kleppmann2015dualwrites,ddia,debezium}. The application
writes its application records and an outbox event into the \emph{same}
local database transaction. A separate relay process (such as a Debezium connector or outbox
worker) tails the committed database log or outbox table, delivers each event
to the external sink, and records its own progress.

\emph{The relocation of the dual write.}
Change Data Capture relocates the dual write rather than eliminating it.
Instead of occurring in the application handler, the dual write now sits
in the relay delivery loop. For every event, the relay must execute two
distinct durable operations under independent authorities: first delivering
the event to the external sink, and second recording a local checkpoint (such
as advancing an offset, watermark table, or message cursor). Because no
transaction spans the external sink and the relay local checkpoint store, a
crash between delivery and checkpointing reopens the dual-write dilemma.
When both source and sink are Kafka topics and consumer offsets are committed
inside the producer transaction via \texttt{sendOffsetsToTransaction}, the
system achieves a shared commit where the dual-write window never opens
(Section~\ref{sec:related}). In heterogeneous pipelines,
Theorem~\ref{thm:image} (Section~\ref{sec:store}) shows that store-tier
crash safety amounts, under its assumptions, to one condition, that the
derived copy matches the committed log at every running frontier, and the
delivery obligation then moves to the relay (Section~\ref{sec:store}).

\paragraph{Defining the acceptance boundary for API sinks and email.}
To design recoverable delivery, engineers must draw a clear boundary
around what the delivery stage controls. Consider an application sending
shipping confirmation emails through an external provider API (such as
Mailgun or AWS SES), operating across two distinct regimes:
\begin{itemize}
\item \emph{Durable API acceptance.} The delivery obligation is satisfied
  when the provider durably accepts custody of the payload. An HTTP 200 OK
  or 202 Accepted response containing a provider-assigned message ID evidences
  acceptance only if issued after durable persistence, as documented by
  the provider. If the provider does not document post-persistence acknowledgement,
  conservative recovery must treat the response as non-durable and verify
  delivery via the provider query API. In the formal model
  (Section~\ref{sec:model}), this transition marks durable acceptance: once
  accepted, the obligation is completed regardless of downstream transit.
\item \emph{Downstream asynchronous transit.} Everything following durable
  API ingestion (internal queues, MTA routing, SMTP handshakes, spam scoring,
  and inbox polling) lies outside the delivery stage's authority. Transit
  delays must not be conflated with delivery acceptance.
\end{itemize}
This distinction applies equally to payment gateways (such as Stripe charges~\cite{stripe}),
message brokers, and webhooks. A relay crashing while awaiting an HTTP response
cannot deduce from local state whether the sink accepted the payload before the
connection dropped.

\paragraph{Why local state cannot prove delivery.}
When a relay crashes and restarts, its local log and outbox records are
intact, and the last durable checkpoint cursor points to event $N$. Whether the
sink accepted event $N+1$ before the crash or the crash occurred before
transmission cannot be deduced from local state alone. If the sink accepted
event $N+1$ immediately before the crash, resending creates a duplicate. If the
crash occurred before transmission, skipping drops committed work.

This ambiguity is an information limit. A local checkpoint records
\emph{the relay's progress in persisting its own progress records}, not the
sink's durable acceptance. In a post-crash state where delivery succeeded
but the checkpoint write failed, the source database, write-ahead log,
outbox table, and cursor look bit-for-bit identical to a state where
delivery never occurred. Any source-measured recovery policy must choose the
same action in both states, duplicating delivery in one or silently omitting
it in the other (Theorems~\ref{thm:obsbound}-\ref{thm:controlplane},
Section~\ref{sec:wall}).

Committing a checkpoint after every single message reduces the replay
backlog during recovery, but does not eliminate the crash window. It merely
cycles the window once per event. Theorem~\ref{thm:checkpoint}
(Section~\ref{sec:checkpoint}) strengthens the bound to this exact setting:
for a deterministic protocol that checkpoints after every send, crash timing
alone constructs two reachable post-crash states that agree on the local cursor
and disagree on sink acceptance.

\paragraph{Classifying sinks by recovery capability.}
Recovery depends on two distinct capabilities exposed by the sink,
namely what the sink lets recovery \emph{read} (visibility) and what the sink
enforces at \emph{admission} (control). Practical sinks fall into three
archetypes (summarized in Table~\ref{tab:practitioner-sinks}).

\begin{table*}[t]
\caption{Sink capability archetypes and corresponding recovery contracts.}
\label{tab:practitioner-sinks}
\footnotesize
\begin{tabular}{@{}>{\raggedright\arraybackslash}p{0.18\textwidth}>{\raggedright\arraybackslash}p{0.24\textwidth}>{\raggedright\arraybackslash}p{0.26\textwidth}>{\raggedright\arraybackslash}p{0.26\textwidth}@{}}
\toprule
\textbf{Sink archetype} & \textbf{Examples} & \textbf{Read and admission capability} & \textbf{Recovery contract} \\
\midrule
Blind or write-only &
Basic HTTP webhooks, write-only queues, legacy email APIs &
Cannot query past accepted events, with no admission deduplication or fencing tokens. &
No source-measured policy can guarantee exactly-once delivery at the recovery frontier. The system must explicitly choose between duplicate retries or omission under ambiguity. \\
\addlinespace[3pt]
Queryable log or store &
Kafka topic (with read ACLs), replicated datastore table, queryable event store &
Recovery can inspect sink history to determine which source events are missing, but the sink does not natively deduplicate. &
Compute difference between committed source log and sink records, and resend only missing events (Theorem~\ref{thm:escape}). Requires strictly ascending coordinates (P4) and generation fencing against stragglers and concurrent recoverers. \\
\addlinespace[3pt]
Deduplicating or fenced &
Stripe API (idempotency keys), Kafka transactional producer (\texttt{transactional.id} + epoch), cloud queues with FIFO deduplication &
Sink admission gate deduplicates by client-supplied key or rejects stale producer epochs. &
Deduplicate by stable event key, or compute missing events from authoritative acceptance records and use fenced recovery. Retain deduplication state throughout the recovery horizon. \\
\bottomrule
\end{tabular}
\end{table*}

\emph{Archetype 1: Blind or write-only sinks.}
When the sink API provides no query endpoint to inspect past acceptances
and does not support idempotency keys or fencing, recovery cannot determine
what was delivered. No recovery policy that consults only source-side state
can achieve exactly-once delivery at the recovery frontier
(Theorem~\ref{thm:obsbound}). The system must explicitly select its failure
preference by either resending on ambiguity (accepting duplicates) or skipping on
ambiguity (accepting omissions). Client-side retries and backoff do not remove this
trade-off~\cite{treat}.

\emph{Archetype 2: Queryable logs and stores.}
When writing to a sink whose accepted history can be read back (such as a
Kafka topic or secondary database table), recovery can query the sink to
inspect which events are present. By computing the difference between the
committed source log and the sink's accepted records, recovery resends only
missing events (the Sink-Reading Escape of Theorem~\ref{thm:escape},
Section~\ref{sec:door}). Three operational conditions govern this approach:
\begin{itemize}
\item \emph{Coordinate monotonicity (Condition P4).} Source coordinates must
  be strictly increasing per event in commit order (such as a database WAL LSN
  or an outbox event coordinate assigned at commit). Auto-increment sequence
  numbers generated at insert time can commit out of order under concurrent
  transactions, which breaks the commit-order assumption the frontier
  relies on and can leave committed events below a frontier unrecovered.
\item \emph{The credential constraint.} Production services provisioned with
  write-only permissions (such as Kafka produce rights without consume
  rights) silently degrade an Archetype~2 sink into an Archetype~1 blind sink.
\item \emph{Acceptance history versus current state.} In a key-value store,
  reading current state (such as \texttt{status = "shipped"}) distinguishes
  some crash states, but fails if intermediate updates were overwritten.
  Recovery requires a record of per-operation acceptance. Compaction that
  preserves key presence is safe (Section~\ref{sec:memory}), but
  presence-destroying retention (such as hard-deleting records or purging
  tombstones) destroys queryable evidence. Sinks read from replicas must also
  enforce strong read consistency to avoid replica lag hiding recent writes.
\end{itemize}

\emph{Archetype 3: Deduplicating and fenced sinks.}
Stripe deduplicates requests using client-supplied idempotency keys~\cite{stripe}.
Kafka's stable
\texttt{transactional.id} supports transaction recovery and fences old
producers~\cite{kip98}. A relay can still resend an event already committed
to Kafka. It must also commit its progress with the output
or use the sink's authoritative acceptance record and the fencing rules below
to send only missing events.
With permanent deduplication memory and strictly increasing coordinates,
at-least-once delivery gives exactly-once
coverage (Theorem~\ref{thm:dedup}, Section~\ref{sec:memory}). In practice,
finite memory requires a retention window long enough to cover replay.
Derive idempotency keys from stable source coordinates and keep them unchanged
across retries. Payload hashes alone can merge distinct operations with identical
contents and suppress required delivery.

\paragraph{Controlling in-flight stragglers and concurrent recoverers.}
Even with queryable or deduplicating sinks, two distinct distributed systems
hazards can invalidate recovery.

\emph{Hazard 1: In-flight network stragglers.}
Suppose Relay~1 sends an event to a queryable sink, but crashes while the
packet remains buffered in a retrying network proxy or transport queue.
A recovery worker restarts, queries the sink, observes that event $N+1$ is
missing, and resends event $N+1$, which the sink accepts. Moments later, the
delayed packet from Relay~1 arrives at the sink. If the sink lacks arrival
fencing, it accepts the delayed request, causing a duplicate.
Theorem~\ref{thm:wire} (Section~\ref{sec:wire}) proves that no channel-blind
policy reading only durable state can detect in-flight stragglers. The proved
remedy is the atomic fenced re-drive (Theorem~\ref{thm:fence}), where the sink
atomically accepts the recovery delta and advances its generation threshold to
one above the crashed run's generation, rejecting older requests at
admission. In this model, recovery uses the crashed generation. Fencing first
blocks its requests, while fencing afterward leaves a window for stragglers.
Recovery with a new generation falls outside this argument
(Section~\ref{sec:wire}).
\emph{The cost of arrival fencing.} While arrival fencing guarantees
exactness at the recovery frontier, Corollary~\ref{cor:price} proves that it
can reject delayed requests that carried work needed at subsequent frontiers,
leaving that subsequent work as a separate unfulfilled obligation.

\emph{Hazard 2: Concurrent recoverers.}
In containerized environments, a recovery worker losing its heartbeat can
cause an orchestrator to launch a replacement worker. If both workers run
concurrently, both query the sink, observe event $N+1$ missing, and fire
re-drive batches.
Distributed lock leases do not prevent this hazard~\cite{klepp-lock}. If a
worker pauses and its lease expires, a replacement worker claims leadership,
but the original worker can wake and fire its prepared network call.
Theorem~\ref{thm:second} (Section~\ref{sec:claim}) proves this
Second-Recoverer Bound, establishing that without an admission fence,
concurrent recoverers can duplicate deliveries or lose work. The proved
mitigation is atomic claim fencing (Lemma~\ref{lem:claim} and
Theorem~\ref{thm:claim}), under which a recoverer atomically sets the sink fence
to its own generation and arms its batch from the sink record read under that
claim. The fire step succeeds at the sink only if the fence has not been
advanced past the worker's generation.
Theorem~\ref{thm:claim} requires window freshness (Condition C3, requiring that no
source commit occurs within the recovery range between claim and fire) and
validates generation authority at the sink admission gate (Condition C4).

\emph{Fencing internal versus external SaaS sinks.}
Kafka enforces epoch fencing natively at admission via transactional producer
epochs. Relational datastores implement the fence as a conditional generation
check inside the same transaction that applies the re-drive batch. External
SaaS APIs lacking epoch endpoints (such as Stripe) cannot enforce
claim fencing directly. Source outbox leases reduce redundant calls under normal
operation, but cannot stop a resumed worker from firing. For external APIs,
admission deduplication via deterministic idempotency keys is the sole barrier
absorbing stragglers. As Section~\ref{sec:claim} notes, no theorem merges
arrival and claim fencing: they address distinct failure models on different
machines.

\paragraph{Retention horizons across deduplication caches and source logs.}
Deduplication state cannot be retained indefinitely in real systems.
\begin{itemize}
\item Stripe's documentation specifies that idempotency keys become
  eligible for removal after at least 24 hours~\cite{stripe}.
\item Apache Kafka's producer ID expiration
  (\texttt{producer.\allowbreak id.\allowbreak expiration.ms}) defaults to 1 day~\cite{kafka-exp},
  though topic log retention can cause earlier expiry.
\item Cloud FIFO queues configure sliding deduplication windows (such as
  a 5-minute window from message acceptance).
\item Upstream database outbox tables undergo WAL or binlog purging.
\end{itemize}

\emph{Replay after retention expiry.}
If an upstream database undergoes a prolonged outage or network partition
and recovery replays historical events after the sink's deduplication cache
has expired, the sink treats the replayed requests as novel transactions,
triggering a wave of duplicate side effects.

\emph{Sizing and retention guidance.}
\begin{itemize}
\item \emph{Sink deduplication horizon.} The theoretical model in
  Section~\ref{sec:memory} measures deduplication memory in entry-count
  horizons. In practice, converting duration to entry capacity requires
  bounding the maximum arrival rate. The retention horizon must exceed the
  longest anticipated recovery window:
  $T_{\mathrm{dedup}} > T_{\mathrm{max\_outage}} + T_{\mathrm{operator\_delay}}
  + T_{\mathrm{replay\_catchup}} + \mathrm{margin}$.
  If deduplication evidence expires before replay finishes, the sink admits
  replayed events as novel (Section~\ref{sec:memory}). When fixed provider
  windows (such as a 5-minute FIFO deduplication window) cannot be widened,
  replays exceeding the window cannot rely on the sink for deduplication.
  If the queue cannot report past acceptances, recovery after expiry is blind,
  as in Archetype~1. An auxiliary record helps only if it is authoritative and
  updated atomically with queue acceptance. Separate sender-side bookkeeping
  leaves the crash window open.
\item \emph{Source history retention.} To maintain journal-grade exactness,
  upstream outbox records must be retained until the downstream sink has
  durably confirmed their acceptance (Corollary~\ref{cor:journal}).
  Prematurely truncating source history removes the content needed to identify
  which operations require delivery, leaving only coordinate bounds or metadata.
  With the historical content lost, recovery is forced to choose between phantom
  fabrication and abandonment (Theorem~\ref{thm:trunc}).
\end{itemize}

\paragraph{Architecture checklist, auditing, and recovery tests.}
To translate these boundaries into practice, teams can apply the following
checklist, monitoring strategy, and recovery tests.

\emph{Architecture checklist.}
\begin{itemize}
\item \emph{Monotonic event coordinates.} Assign stable event IDs at source
  commit time strictly increasing in commit order (such as a database WAL LSN
  or an outbox event coordinate assigned at commit). Auto-increment sequence
  numbers generated at insert can commit out of order under concurrent transactions.
  Never generate fresh random UUIDs during retries.
\item \emph{Documented sink contracts.} Explicitly classify each sink as
  Blind, Queryable, or Deduplicating. For Blind sinks, record the failure
  preference (duplicates versus omissions).
\item \emph{Arrival fencing for in-flight stragglers.} For internal sinks,
  ensure recovery atomically raises the generation fence upon re-drive,
  preventing delayed in-flight requests from prior runs from landing.
\item \emph{Claim fencing for concurrent recoverers.} Ensure replacement
  recovery workers claim leadership and validate generation tokens at the
  sink admission gate before firing batches.
\item \emph{Verified read permissions.} Ensure relay service accounts
  possess explicit read/consume permissions where queryable recovery is
  assumed.
\item \emph{Dead-letter queue trade-off.} Routing unprocessable payloads
  to a Dead-Letter Queue (DLQ) prevents head-of-line blocking, but accepts
  a frontier omission. The reconciliation monitor must record these coordinates
  as acknowledged omissions (avoiding persistent alerts), and the DLQ write
  itself represents an auxiliary dual-write boundary.
\item \emph{Outbox storage safeguards.} Because outbox records must be
  retained until downstream confirmation (Corollary~\ref{cor:journal}),
  extended sink outages cause outbox tables to accumulate rows. Implement
  threshold alerts on row count and WAL utilization to prevent storage
  exhaustion.
\end{itemize}

\emph{Monitoring and reconciliation auditing across frontiers.}
Delivery failures are asymmetrical in production. Duplicate side effects are
loud (prompting complaints or billing alerts), whereas omissions are silent
(a skipped event can leave dashboards green).
To detect silent omissions on Queryable or Deduplicating sinks, implement an
independent \emph{reconciliation auditor}, a background job querying the source
database for committed event coordinates and verifying their presence in the
sink's acceptance log between the last verified position and the latest completed
frontier. Alert whenever a committed obligation remains unaccepted past the SLA.

\emph{Recovery tests.}
Dual-write and CDC pipelines should be verified in staging against failure
scenarios tailored to sink archetypes:
\begin{itemize}
\item \emph{The checkpoint crash (all archetypes).} Terminate the relay
  immediately after the sink ingests a request but before the local checkpoint
  commits. For Archetype~1 (Blind), verify the system produces the configured
  compromise (duplicate retry or skip). For Archetypes~2 and~3, verify that
  restart does not create duplicate effects.
\item \emph{The in-flight straggler (Archetype~2).} Inject artificial
  latency on a delivery request, trigger a timeout and recovery restart, let
  recovery complete, and then release the delayed original request. Verify
  that the sink rejects the straggler at its generation fence.
\item \emph{The split-brain worker (Archetypes~2 and~3).} Pause an active
  recovery worker, allow a failover worker to claim leadership
  and complete recovery, then resume the original worker.
  Verify that the superseded worker's writes are rejected by the claim fence.
\item \emph{The expired deduplication replay (Archetype~3).} Fast-forward cache
  expiry or flush the sink's deduplication state, then trigger a manual replay
  of historical events. Verify that operational safeguards reject blind
  replays across expired retention windows.
\item \emph{Crash inside fenced recovery.} Inject a crash during recovery
  re-drive. Because the model abstracts fenced re-drive as an atomic transition
  (Section~\ref{sec:mech}), implementations execute distinct steps (reading
  accepted history, submitting re-drive, and raising the fence). Verify that a
  crash before fence advancement leaves the system recoverable without
  duplicate side effects.
\end{itemize}

\section{Conclusion}\label{sec:conclusion}

In this paper, we developed a machine-checked theory of recovery when a
process delivers an effect to one system and records its progress in
another. We proved that the durable state of the crashed side cannot
decide whether a delivery should be repeated. Two reachable states can
share identical source databases, logs, control states, restart
counters, and checkpoints while differing in what the sink accepted. Any policy
consulting only this shared state must duplicate an effect in one
execution or omit committed work in the other. We also proved
this result for a deterministic deliver-then-checkpoint protocol in which
crash timing is the only nondeterminism.

The positive results identify the controls required for exact recovery.
An authoritative, complete, and current record of sink acceptance allows
recovery to send only the committed operations the sink has not yet
accepted. This works only if nothing changes the record between the read
and the re-send. For requests still in flight from the crashed run,
recovery uses an arrival fence. For overlapping recoverers, an atomic
claim ties the read to the worker that made it, and the sink accepts the
batch only while that claim is still the latest. These mechanisms solve distinct
hazards on different machines, with conditions and costs that form part of the
guarantee.

The guarantee also depends on evidence retention. With permanent
deduplication, every required payload appears once in the deduplicated
view exactly when each has been delivered at least once. A
finite deduplication window can admit replays as new after identity
expiry, and truncating source history removes the information needed to
identify undelivered operations. Practical recovery must therefore retain
sink acceptance evidence and source history for the full recovery horizon.

Transactional outboxes and change data capture remain the right tools
for deriving downstream state from a commit log. At every recovery
source-log position, downstream state must match the state obtained by
replaying the committed log to that position. This agreement is necessary
and sufficient for store-level crash safety under the model's assumptions.
The relay created by that design has its own delivery boundary. Reviews of
that boundary should identify the sink's durable acceptance event, active
requests or recoverers, and retention horizons on both sides. A source
checkpoint tracks progress, but cannot replace the sink's acceptance record.

\paragraph{AI assistance.}
Generative AI assistants \mbox{Claude Fable~5} (Anthropic) and
\mbox{ChatGPT~5.6 Sol} (OpenAI) supported drafting, revision, and
Isabelle/HOL formalization. The author directed the research, verified all
proofs and text, and takes full responsibility for this work. The assistants
are tools, not authors.

\bibliographystyle{unsrtnat}
\bibliography{refs}
\flushcolsend

\end{document}